\documentclass[12pt]{article}
\usepackage{amsmath}
\usepackage{caption,subcaption,graphicx,psfrag,epsf,cases,empheq}
\usepackage{enumerate}
\usepackage{blkarray}
\usepackage{stackrel}
\usepackage{url}
\usepackage[top=1in, bottom=1in, left=1in, right=1in]{geometry} 
\usepackage{caption,subcaption,graphicx,psfrag,epsf,cases,empheq,color,amssymb,amsmath,graphicx,float,booktabs,amsthm,bm,bbm,multirow,threeparttable,textcomp,adjustbox,setspace,natbib,amsthm,soul,cancel,mathrsfs,euscript,amsfonts,xcolor,tikz,bm,blkarray,stackrel}
\usepackage{amssymb, amsmath}
\usepackage[ruled,vlined,linesnumbered]{algorithm2e}
\usepackage{caption}
\usepackage{subcaption}
\usepackage{array}
\usepackage{booktabs}
\newtheorem{definition}{Definition}
\newtheorem{assumption}{Assumption}
\newtheorem{theorem}{Theorem}
\newtheorem{lemma}{Lemma}

\newcommand{\ind}{\perp\!\!\!\perp}

\usepackage{pifont}

\usepackage{pgf,tikz,wrapfig}
\usetikzlibrary{arrows,shapes.arrows,shapes.geometric,shapes.multipart,decorations.pathmorphing,positioning}
\tikzset{
	>=stealth',
	true/.style={
		rectangle,
		draw=black, very thick,
		text width=6.5em,
		minimum height=2em,
		text centered,
		fill=gray, opacity = 0.5},
	punkt/.style={
		rectangle,
		rounded corners,
		draw=black, very thick,
		text width=6.5em,
		minimum height=2em,
		text centered},
	est/.style={
		circle,
		draw=black, very thick,
		text centered},
	shade/.style={
		circle,
		draw=black, very thick, fill=gray!50,
		text centered},
	weight/.style={
		circle,
		draw=black, very thick,
		text width=6.5em,
		minimum height=2em,
		text centered},
	pil/.style={
		->,
		thick,
		shorten <=2pt,
		shorten >=2pt,},
	double/.style={
		<->,
		very thick,
		shorten <=2pt,
		shorten >=2pt,},
	dash/.style={
		dashed,
		very thick,
		shorten <=2pt,
		shorten >=2pt,},
	dashdouble/.style={
		<->,
		dashed,
		thick,
		shorten <=2pt,
		shorten >=2pt}}

\date{}
\begin{document}
	\def\spacingset#1{\renewcommand{\baselinestretch}%
		{#1}\small\normalsize} \spacingset{1}
	%%%%%%%%%%%%%%%%%%%%%%%%%%%%%%%%%%%%%%%%%%%%%%%%%%%%%%%%%%%%%%%%%%%%%%%%%%%%%%
	\title{\Large\bf Data-driven Automated Negative Control Estimation (DANCE): Search for, Validation of, and Causal Inference with Negative Controls}
	\author{Erich Kummerfeld$^1$, Jaewon Lim$^2$, and Xu Shi$^3$ \\
		$^1$Institute for Health Informatics, University of Minnesota
		\\$^2$Department of Biostatistics, University of Washington
		\\$^3$ Department of Biostatistics, University of Michigan}
	\maketitle
	
	\bigskip
	\begin{abstract}
	Negative control variables are increasingly used to adjust for unmeasured confounding bias in causal inference using observational data. They are typically identified by subject matter knowledge and there is currently a severe lack of data-driven methods to find negative controls. 
In this paper, we present a statistical test for discovering negative controls of a special type---disconnected negative controls---that can serve as surrogates of the unmeasured confounder, and we incorporate that test into the Data-driven Automated Negative Control Estimation (DANCE) algorithm. DANCE first uses the new validation test to identify subsets of a set of candidate negative control variables that satisfy the assumptions of disconnected negative controls.  It then applies a negative control method to each pair of these validated negative control variables, and aggregates the output to produce an unbiased point estimate and confidence interval for a causal effect in the presence of unmeasured confounding. We (1) prove the correctness of this validation test, and thus of DANCE; (2) demonstrate via simulation experiments that DANCE outperforms both naive analysis ignoring unmeasured confounding and negative control method with randomly selected candidate negative controls; and (3) demonstrate the effectiveness of DANCE on a challenging real-world problem.
	\end{abstract}
	
	\vfill
	
	\newpage
	\spacingset{1.45} 
	
\section{Introduction}
There are many causal questions in science and medicine that can not be solved with randomized experiments now or in the foreseeable future. For such questions, our best estimates must thus rely on observational data instead. The rich field of causal inference has developed in response to this, providing support for these efforts and developing methods that offer some level of assurance and confidence for learning causal information from observational data \citep{pearl2009causality,rubin1974estimating}. 
	Many causal inference methods assume that there are no unmeasured common causes of treatment and outcome, but it is generally believed that in reality unmeasured confounders are widespread. This is a serious limitation to the methods that make such assumptions. One of the most frequently used approaches to mitigate unmeasured confounding is the instrumental variable (IV) approach \citep{angrist1991does,angrist1996identification,hernan2006instruments}, which has been previously studied extensively  \citep{greenland2000introduction,baiocchi2014instrumental,garabedian2014potential,burgess2017review,swanson2018partial}.

	A more recently developed strategy is {negative control} (NC)  methods \citep{lipsitch2010negative,shi2020selective,tchetgen2020introduction}.
	Negative controls are  variables associated with the unmeasured confounders but not causally related to either the treatment or outcome variables of primary interest.
	One can detect residual confounding bias leveraging such known null effects: 
	presence of an association between the negative control and the exposure or outcome  constitutes compelling evidence of residual confounding bias, while the absence of such association implies no empirical evidence of such bias. 
	NCs have traditionally been used to rule out non-causal explanations  of empirical findings
	\citep{rosenbaum1989role,weiss2002can,lipsitch2010negative,glass2014experimental}. Recently, a sequence of NC methods have been developed to identify causal effects and correct for unmeasured confounding bias  \citep{miao2018identifying,deaner2018proxy,shi2018multiply,singh2020kernel,cui2020semiparametric,ying2021proximal,kallus2021causal,dukes2021proximal,li2022doubly}. 
	
	A key challenge in the use of NC methods is that until now, NC variables have had to be identified laboriously from background knowledge.	It also had to be assumed that the identified variables were genuine NCs, as no validation test existed unless one is willing to make additional assumptions. Such situations are common in causal inference, e.g., the assumption of no unmeasured confounding is also untestable. Nevertheless, we will show that under certain conditions, it is possible to leverage certain subcovariance matrix rank constraints to validate a particular class of NC variables, referred to as disconnected NCs which we formally define in Section~\ref{sec:negcontrol}, satisfying a specific causal structural model.
	
	In this paper, we utilize some lesser known theory regarding relationships between subcovariance matrix rank constraints and the graphical structure of causal models to provide both theory and algorithms for evaluating NC variables. First, we provide a statistical test that can be used to determine whether a triplet of candidate NCs are real disconnected NCs or not. Second, we provide a simple algorithm for searching among a set of candidate NCs, and identifying subsets of those variables that collectively meet the conditions of being disconnected NCs. Third, we combine our proposed method for finding valid NC variables with a recently developed double-NC method for causal inference \citep{miao2018confounding,shi2018multiply,cui2020semiparametric}, creating an algorithm that accurately estimates and makes inferences about causal effects from observational data. We refer to the proposed method as the Data-driven Automated Negative Control Estimation (DANCE) algorithm. We prove that our proposed methods are correct under fairly general assumptions, evaluate their  finite sample performance with a series of numerical experiments, and demonstrate their usability on a real world data set.
	
	The rest of the paper is organized as follows. In Section \ref{background} we review the three main topics that the work in this paper builds upon: negative controls, structural models, and rank constraints. We then present a statistical validation test for disconnected NCs in Section \ref{validation}, and prove its correctness in Section \ref{DNCTcorrectness}. Section \ref{FNCalgo} presents an algorithm that searches a set of candidate NC variables to find sets of disconnected NCs which pass the validation test, and Section \ref{DNCcausalinference} presents the DANCE algorithm that combines with the double-NC method to construct an all-in-one method for producing a valid causal effect estimate from a data set containing a collection of candidate NC variables, some of which are not necessarily valid disconnected NCs. Section \ref{simulations} presents numerical experiments to evaluate our proposed test and algorithms, and compares them to two methods: a simple regression method ignoring unmeasured confounding and  a random selection of candidate NCs followed by the double-NC method. An application of DANCE to a real clinical data set is described in Section \ref{demonstrations}. Section \ref{Discussion} summarizes the strengths and limitations of the methods presented in this paper, and points towards promising directions for future work.
	
	\section{Background}
	\label{background}

	\subsection{Unmeasured Confounding and Negative Control Methods}
	\label{sec:negcontrol}
	We adopt the potential outcome framework under the Stable Unit Treatment Value Assumption (SUTVA) \citep{rubin1974estimating,rubin1980discussion,cox1992planning} and let $(O(1),O(0))$ denote the pair of potential outcomes under treatment and control conditions, respectively. We are interested in estimating the average treatment effect (ATE), defined as $\Delta=E[O(1)-O(0)]$. It suffices to identify the counterfactual mean $E[O(t)]$ for $t\in \{0,1\}$.
	Let $O$ denote the observed outcome and $T$ denote the binary treatment. 
	We suppress measured covariates for simplicity; adjustment for measured covariates is discussed in Section \ref{subsec:doublenc}.
	%We assume the standard assumption of consistency, i.e., the observed outcome $O=O(t)$ if $T=t$ for $t\in \{0,1\}$. 
	
	Instead of making the no unmeasured confounding assumption, we allow the presence of an unmeasured confounder $U$ with a latent ignorability assumption that $O(t)\ind T\mid U$. If $U$ was measured, then $E[O(t)]$ is identified under the ignorability assumption \citep{robins1986new}. However, when $U$ is unobserved and unadjusted, ATE estimation will be biased. In this case, additional information is needed to identify and make inference about the ATE. 

	An increasingly popular approach to mitigate bias due to unmeasured confounding is to use its proxies. For example, as shown in Figure~\ref{fig:DAG}, 
	if $U$ can be measured with error via proxy variables $Z$ and $W$, then one can leverage $Z$ and $W$ to identify the confounding bias due to $U$ and remove such bias from the estimated causal effect. Such proxy variables have been referred to as negative controls \citep{lipsitch2010negative,shi2020selective}. Formally, a negative control outcome, denoted as $W$, is a variable known not to be causally affected by the treatment of interest. Likewise, a negative control exposure, denoted as $Z$, is a variable known not to causally affect the outcome of interest. The negative control exposure and outcome variables should share a confounding mechanism with the exposure and outcome variables of primary interest. In summary, $Z$ and $W$ satisfies 
	\begin{equation}
	    (T,Z)\ind (O(t),W)\mid U. \label{eq:NC}
	\end{equation}
	There are a number of causal graphs that satisfy the NC assumptions \citep{shi2020selective}. For example, both a valid instrumental variable independent of the unmeasured confounder and an invalid instrumental variable associated with the unmeasured confounder are valid negative control exposures.  Alternative directed acyclic graphs encoding the NC assumptions are available in \cite{shi2020selective}.

Figure~\ref{fig:DAG} presents a special case where $Z$ and $W$ are causally related to neither the treatment nor the outcome of interest, hence $Z$ and $W$ can serve as either negative control exposure or negative control outcome \citep{shi2018multiply,tchetgen2020introduction}. 
We refer to such a special class of NC variables as the \emph{disconnected NCs}. Formally, the disconnected NCs satisfy the following assumption
\begin{equation*}
    (Z,W)\ind (T,O)\mid U.
\end{equation*}
Compared to the fundamental NC assumption \eqref{eq:NC}, the disconnected NCs satisfy additional assumptions that $Z\ind T\mid U$ and $W\ind O\mid U$.

\begin{figure}[h]
	\centering
	\resizebox{.5\textwidth}{!}{%
		\begin{tikzpicture}
			\node[est]  (A) at (0.5,-1.3) {T};
			\node[est] (Y) at (3.5,-1.3) {O};
			\node[shade] (U) at (2,0.5-0.3) {U};
			\node[est] (Z) at (5,0.5-0.3) {Z};
			\node[est] (W) at (-1,0.5-0.3) {W};
			%	\node[est] (IV) at (-2.5,-1.5) {IV};
			\foreach \from/\to in {A/Y,U/A,U/Y,U/W,U/Z}\path[->,pil,black] (\from) edgenode {} (\to);
			%	\draw (Z) -- (A);
			%%%%
			\coordinate [label={treatment}] (treatment1) at (0.5,-2.15);
			%%%%
			\coordinate [label={outcome}] (outcome) at (3.5,-2.15);
			%%%%
			\coordinate [label={unmeasured}] (confounders) at (2,0.6+0.5-0.15);
			\coordinate [label={confounders}] (confounders) at (2,0.35+0.4-0.15);
			%%%%
			\coordinate [label={negative}] (NCE0) at (5,-0.65-0.1);
			\coordinate [label={control}] (NCE0) at (5,-0.9-0.1);
			%%%%
			\coordinate [label={negative}] (NCO0) at (-1.05,-0.66-0.1);
			\coordinate [label={control}] (NCO0) at (-1.05,-0.95-0.1);
		\end{tikzpicture} %
	}
	\caption{\label{fig:DAG} Causal graph of two disconnected NCs, Z and W, suppressing the measured covariates X which is implicitly conditioned on in all arguments.}
\end{figure}
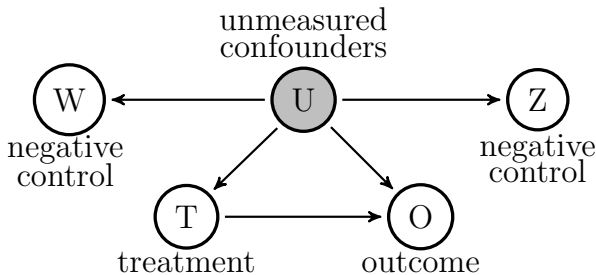

Using a pair of negative control exposure and outcome variables, referred to as the double-NC, \cite{miao2018identifying} established nonparametric identification of the average treatment effect (ATE). 
Intuitively, having additional children of $U$ that are conditionally independent with $T$ and $O$ allows for identification of the unmeasured confounding bias due to the influence of $U$ on $T$ and $O$, and subsequently this quantity can be removed from the association between $T$ and $O$, leaving an unbiased estimate of $T$'s effect on $O$. Recently, the NC framework has been extended to proximal causal inference, which partitions measured covariates into proxies satisfying NC conditions, acknowledging that covariate measurements are at best proxies of the underlying confounding mechanisms \citep{tchetgen2020introduction,cui2020semiparametric,ying2021proximal,dukes2021proximal}.

	\subsection{Structural Models}\label{StructuralModels}
	The rest of this paper makes heavy use of the theory of structural graphical models, especially structural equation models. This section provides background terminology and definitions from this field, which can be used as an introduction for readers who are unfamiliar with it, or as a reference for readers who are already familiar with it.
	
	A \emph{directed graph} is a pair of sets, $\langle \Phi, \Psi \rangle$, where $\Phi$ contains some number of variables, and $\Psi$ contains some number of directed edges, or arrows, pointing from one variable in $\Phi$ to another variable in $\Phi$. It is often important to consider not just individual edges, but \emph{paths} in a directed graph, which are ordered lists of edges in a graph such that each edge shares one of its endpoint variables with the edges before and after it. In other words, a path is a connected sequence of edges. A special kind of path is a \emph{trek}, where nowhere along the length of the path are there two consecutive edges with arrows pointing to the variable between them. In other words, in a trek there are no subsections that look like this: $\rightarrow X \leftarrow$. A special type of trek is the \emph{directed path}, where all the edges point in the same direction along the path. Probably the most commonly studied subclass of directed graphs is \emph{directed acyclic graphs} (DAGs), which are directed graphs that contain no directed paths passing through the same variable twice, i.e., directed graphs without any directed cycles.
	
	Directed graphs are often augmented with additional quantitative and statistical information, forming a statistical model. Two commonly used examples are \emph{structural equation models} (SEMs) and \emph{Bayesian networks} (BNs). The specific meanings of these terms often vary from one paper or research group to another. In this paper, we consider SEMs to refer to DAGs which have been augmented with an additive Gaussian noise term $\epsilon_X$ for each variable $X$ in the DAG, as well as a linear coefficient $\beta_{VX}$ corresponding to each edge $V \rightarrow X$. As such, each variable's value is calculable from the value of its \emph{parents} (the nodes with edges pointing to it) and the value of its independent noise term. Let $parents(\cdot)$ be a function that returns all of the parents of the variable in its argument, then the SEM can be written as a collection of equations of the form: 
	\[X = \sum_{V \in parents(X)}\beta_{VX} V + \epsilon_X.\]
	
	In this paper, we consider BNs to be similar to the SEM class of models, but using categorical variables instead of continuous variables. Each variable in a BN has its probability distribution determined by a conditional probability table based on the values of that variable's parents. As such, like an SEM, it specifies a complete joint distribution that can be decomposed into separate components for each of its variables. As we are defining them here, a BN specifies a multinomial distribution, while an SEM specifies a multivariate Gaussian distribution. 
	
	Variables can be either \emph{measured}, or \emph{unmeasured}. Unmeasured variables are also called \emph{hidden} variables or \emph{latent} variables. Whether a variable is measured or unmeasured has no effect on the SEM or BN itself, but rather refers to what variables are available in data collected from that SEM or BN. Unmeasured variables are absent from the data entirely, and so all their values are missing. Notoriously, since unmeasured variables are still part of the data generating model, measured variables in the model can be correlated with each other because an unmeasured variable is the parent of both of them. Such an unmeasured variable is considered an \emph{unmeasured confounder} in many contexts. In causal inference from observational data, unmeasured confounders are a potential source of bias, as it is a source of spurious association that can occur even when neither of the two measured variables influences the other in any way.

	\subsection{Rank Constraints and the Vanishing Tetrad Test}
	\label{sec:rankconstraint}
	
	It is well known that the graphical structure $\mathcal{G}$ of a causal structural model over variables $V$ entails constraints among the partial correlations of the observed (measured) variables $M \subseteq V$. Partial correlation constraints have been leveraged to develop causal discovery algorithms, including algorithms that still operate correctly in the presence of unmeasured common causes, such as Fast Causal Inference (FCI) \citep{spirtes2000causation,zhang2008completeness,hyttinen2013discovering}. These methods all depend on theoretical work establishing theorems that systematically relate graphical structures to their implied partial correlation constraints.
	
	A lesser known type of constraint that is also implied by graphical structure is rank constraints on the subcovariance matrix. A subcovariance matrix is a covariance matrix between two sets of variables, $S_1$ and $S_2$, with $S_1,S_2 \subseteq M$. As first proved by \citet{sullivant2010trek}, the rank of each subcovariance matrix will have bounds due to features of the graph $\mathcal{G}$. Let $cov(X,Y)$ be the covariance between two random variables $X$ and $Y$, then the subcovariance matrix for $S_1=\{W,Z\}$ and $S_2=\{T,O\}$ can be written as: $$
	\Sigma_{S_1,S_2}=\Sigma_{\{W,Z\},\{T,O\}}=\begin{pmatrix}
cov(W,O) & cov(W,T)  \\
cov(Z,O) & cov(Z,T)  \\
\end{pmatrix}.
	$$
	Sullivant's theorems state that, for the graph in Figure~\ref{fig:DAG}, since all paths from $W$ and $Z$ to $T$ and $O$ pass through one variable, $U$,  the rank of matrix $\Sigma_{\{W,Z\},\{T,O\}}$ is less than its dimension. This also implies that the determinant of $\Sigma_{\{W,Z\},\{T,O\}}$ is zero, i.e., $cov(W,O)cov(Z,T) - cov(W,T)cov(Z,O) = 0$. This is referred to as a \textit{vanishing tetrad}. The same is not true of the subcovariance matrix for $\Sigma_{\{W,T\},\{Z,O\}}$, since there is an additional path from $T$ to $O$ that does not include $U$. These theorems were then extended by \cite{spirtes2013calculation} to relax some of their linearity and acyclicity assumptions. A growing number of papers have made use of these theorems to identify the presence of unmeasured common causes, and to even make inferences about the causal relationships among unmeasured common causes \citep{kummerfeld2014causal,kummerfeld2016causal,yang2017finding}.
	
	The theorems relating rank constraints to graphical structure make use of some lesser known graph concepts, in particular the concepts of \emph{trek} and \emph{trek-separation}, typically abbreviated as \emph{t-separation}. We here provide those definitions, including some intermediary definitions necessary for defining trek-separation. In a directed graphical model, a trek is a path with no colliders on it, i.e., no variable within the trek is a direct child of the variables both before and after it in the trek. A variable is said to block a trek if it lies anywhere along that trek. Any trek has exactly one \emph{root} node which has no parents in the trek, which can be any node along the trek including the endpoints. Let $t$ be an arbitrary trek from node $A$ to $B$ in directed graph $\mathcal{G}$, with root node $C$. The portion of the trek $t$ from $C$ to $A$ is naturally also a trek, as is the portion of $t$ from $C$ to $B$. Let $D$ be a node along $t$. We say that $D$ blocks $t$ on the $A$-side of $t$ if $D$ lies along the trek from $C$ to $A$, and likewise $D$ blocks $t$ on the $B$-side of $t$ if it lies along the trek from $C$ to $B$. An ordered pair of sets of variables $\langle S_A, S_B\rangle$ is said to \emph{t-separate} variables $A$ and $B$ in graph $\mathcal{G}$ if for every trek $t$ from $A$ to $B$ in $\mathcal{G}$, $t$ is either blocked on the $A$-side by a variable in $S_A$ or on the $B$-side by a variable in $S_B$. For example, in the graph presented in Figure~\ref{fig:SimGraph2},
	for $A=\{Z_4,Z_5\}$, $B=\{Z_6,Z_7\}$, we have $C=U$, and they can be trek separated in multiple ways, including: $S_A=\{U\}$ and $S_B=\{\}$; $S_A=\{\}$ and $S_B=\{U\}$; $S_A=\{U,T\}$ and $S_B=\{U,Z_1\}$; and so on. By comparison, if $A=\{Z_4,Z_6\}$ and $B=\{Z_5,Z_7\}$, then sets like $S_A=\{U\}$ and $S_B=\{\}$ no longer t-separate $A$ from $B$, since $U$ does not block $Z_4$ from $Z_5$ or $Z_6$ from $Z_7$. Additional variables would have to be included in $S_A$ or $S_B$ to block these additional treks.

	In this paper we make use of the extended trek separation theorems of \citet{spirtes2013calculation}, which relate statements about t-separation in the data generating model to rank constraints in the covariance matrix among certain sets of nodes. The presence or absence of rank constraints can be determined from empirical data using statistical tests for vanishing determinants.  \citet{wishart1928sampling} created a statistical test for the null hypothesis that the determinant of a subcovariance matrix is zero, assuming that the relevant portions of the structural model are Gaussian. Alternative tests \citep{bollen1990outlier,bollen1993confirmatory} that relax this distributional assumption tend to be more computationally intensive and appear to have reduced power. Moreover, in practice, even when the distribution is non-Gaussian, the Wishart test often performs well with a large sample \citep{spirtes2013calculation,silva2017learning}. Therefore in our simulations and applications presented later in this paper, we use the Wishart test \citep{wishart1928sampling}.
	
	For sets of variables $S_1$ and $S_2$ such that $|S_1| = |S_2|$, let $D_{S_1,S_2}$ be the determinant of the subcovariance matrix $\Sigma_{S_1,S_2}$ of all the variables in $S_1$ by all the variables in $S_2$. The Wishart test calculates the distribution of the empirically observed value of $D_{\{a,b\},\{c,d\}}$, denoted as $\widehat{D}_{\{a,b\},\{c,d\}}$, under the null hypothesis that the true determinant is zero and the corresponding tetrad vanishes, i.e., $H_0: D_{\{a,b\},\{c,d\}}=0$. 
	Let $N$ be the sample size of the observed data. The estimated variance
	of $\widehat{D}_{\{a,b\},\{c,d\}}$ is
	$\widehat{\sigma}^2 = \{\widehat{D}_{\{a,b\},\{a,b\}}\widehat{D}_{\{c,d\},\{c,d\}}(N+1)/(N-1) -\widehat{D}_{\{a,b,c,d\},\{a,b,c,d\}}\}/(N-2)$ \citep{wishart1928sampling}.
	The Wishart test then computes $$W=\widehat{D}_{\{a,b\},\{c,d\}}/\widehat{\sigma},$$ and forms a p-value for a two-sided significance test based on the asymptotic normal distribution with mean zero under the null. 
	This is referred to as the vanishing tetrad test.
	
	While the algorithms we present in this manuscript do not rely on Gaussian distributed data, the specific test of vanishing tetrad we use here---the Wishart test---does assume a Gaussian distribution. The same asymptotic correctness results can be extended to other distributions by using other tetrad tests, such as the distribution-free tetrad test developed by  \citet{bollen1993confirmatory}. We use the Wishart test because it is simple to compute and appeared to have better power at plausible finite sample sizes during numerical experiments.
	
	\section{The Statistical Validation of Negative Controls}
	\label{validation}
	
	\subsection{Definitions and Assumptions}
	Before stating the validation test for NCs and then proving its correctness, it will be useful to  present some definitions and assumptions. We first define the following simple NC model that involves an unmeasured confounder $U$ and children of $U$.
	\begin{definition}
	\label{def:sncm}
	A \emph{simple NC model} is a structural equation or Bayesian network model, with measured variables ${M}$, unmeasured variables ${U}$, structural causal relationships ${S}$ and causal coefficients ${C}$, that meets the following criteria:
	\begin{enumerate}
		\item \itemsep -3pt $|{U}|=1$, i.e., there is only one unmeasured variable. Let this unique unmeasured variable be denoted by $U$.
		\item $U$ directly causes all elements of ${M}$.
		\item ${M}$ has two distinct privileged variables, denoted by $T$ and $O$. $T$ may cause $O$, but no other variable in ${M}$ directly causes or is caused by $T$ or $O$.
	\end{enumerate}
	See Figures \ref{fig:DAG} and \ref{fig:Graph1} for examples of simple NC models. Let variables in ${M} \setminus \{X,Y\}$ be called candidate NCs  which are not necessarily valid disconnected NCs.
	\end{definition}
	We clarify that although the simple NC model may appear to be limited to one unmeasured confounder $U$, such an unmeasured confounder can be a spectrum of multiple latent variables. For example, suppose $U$ denotes healthcare seeking behavior, a common source of unmeasured confounding bias of concern in vaccine effectiveness studies, then it is likely that $U=\sum_{j}\alpha_j L_j$, where $L_j$ denotes latent variables such as perceptions of illness and treatment, access to healthcare, and insurance coverage, and $\alpha_j$ denotes the corresponding importance of each latent factor.
	
	In this paper, we will focus on identifying a particular class of NC variables, which we refer to as the {disconnected NCs}. %Besides the fundamental NC assumptions, 
	As presented in Figure~\ref{fig:DAG} and defined in Section~\ref{sec:negcontrol}, a disconnected NC is independent of both the treatment $T$ and the outcome $O$ conditional on the unmeasured confounder $U$. 
	That is, any trek between a disconnected NC and $T$ (or $O$) must pass through $U$. Hence a tetrad of three disconnected NCs plus $T$ (or $O$) may be a vanishing tetrad as introduced in Section~\ref{sec:rankconstraint}, while a tetrad of two disconnected NCs plus both $T$ and $O$ will not vanish due to the potential direct path from $T$ to $O$ that does not go through $U$. Motivated by this observation, 
	our validation test leverages the vanishing tetrad test as detailed in the next section and relies on at least three disconnected NCs.
	%Due to the potential direct path from $T$ to $O$ that does not goes through $U$, a tetrad including both $T$ and $O$ will not vanish and does not allow us to identify disconnected NC. To construct tetrads that vanish, we rely on at least three disconnected NCs, which we define below.
	\begin{definition} A \emph{disconnected negative control triplet} (DNCT) is a set of three candidate NCs in a simple NC model such that all treks from one member of the DNCT to another pass through $U$.
	\end{definition}
	Our DNCT definition %not only excludes pathways that connect a disconnected NC to $T$ (or $O$) without passing through $U$, but also 
	excludes pathways connecting any two disconnected NCs without passing through $U$, because such a pathway will lead to a tetrad that does not vanish and will also violate the NC assumption \eqref{eq:NC} when using such a pair of disconnected NCs as double-NC.
	We make the following assumptions necessary for detecting DNCT in a graph $\mathcal{G}$.
	
	\begin{assumption} 
	\label{ass:1}
	The data is generated by a simple NC model with $|{M}| \geq 5$. This implies that there are at least three candidate NCs.
	\end{assumption}	
	
	\begin{assumption} 
		\label{ass:2}
	\emph{Tetrad Faithfulness}. In the data distribution implied by the simple NC model, tetrads vanish only if they are implied to vanish by the structure of the simple NC model. In other words, tetrads do not vanish as an ``accident'' of the model's particular coefficients.
		\end{assumption}
	
	\begin{assumption} 
	\label{ass:SNCMlinear}
	The data is generated by a simple NC model which is linear and acyclic among its measured variables.
		\end{assumption}

	\subsection{The DNCT Validation Test}
	\label{sec:DNCTtest}
	We now introduce our proposed validation test for the DNCT. For a particular treatment and outcome, the DNCT validation test determines whether a set of three candidate NCs is a DNCT, i.e., a triplet of disconnected NC variables. It serves as a validation test for determining whether a proposed set of candidate NCs meets the assumptions necessary for causal inference with NCs. This test takes the following as input:
	\begin{enumerate}\itemsep -3pt
		\item A table of data, $Data$;
		\item A variable $T$ in the data identified as the treatment;
		\item A variable $O \neq T$ in the data identified as the outcome;
		\item Three candidate NCs in the data which cannot include $T$ or $O$;
		\item A vanishing tetrad test, $VanTetrad(S_1,S_2,Data)\rightarrow \{TRUE,FALSE\}$, for two sets $S_1$ and $S_2$ each containing two variables;
		\item A hyperparameter value between 0 and 1, \emph{alpha}, used by the vanishing tetrad test.
	\end{enumerate}
	
	The output of this test is a true/false Boolean value: true indicates that the set of three candidate NCs forms a DNCT for $T$ and $O$, and false indicates that the set of three candidate NCs is not a DNCT for $T$ and $O$. Our DNCT validation test requires that six specific vanishing tetrad tests must not reject their corresponding null hypotheses that the determinant of a subcovariance matrix is zero, i.e., the corresponding tetrads vanish. More specifically, let $X$, $Y$, and $Z$, be the three candidate NCs being tested. Our test returns $TRUE$ if and only if the following vanishing tetrad tests all accept their null hypotheses that the corresponding tetrads vanish:
	\begin{enumerate}\itemsep -3pt
		\item $VanTetrad(\{X,Y\},\{Z,T\},Data)$;
		\item $VanTetrad(\{X,Z\},\{Y,T\},Data)$;
		\item $VanTetrad(\{Z,Y\},\{X,T\},Data)$;
		\item $VanTetrad(\{X,Y\},\{Z,O\},Data)$;
		\item $VanTetrad(\{X,Z\},\{Y,O\},Data)$;
		\item $VanTetrad(\{Z,Y\},\{X,O\},Data)$.
	\end{enumerate}
	
	These six tetrad tests cover all combinations where exactly one of $T$ or $O$ is included in the tetrad. Intuitively, this tests whether there is more than one pathway connecting any of the candidate NCs to each other or to $T$ or $O$.
	
	\subsection{Correctness of the DNCT Validation Test}
	\label{DNCTcorrectness}
	We have the following theorem which states that the DNCT validation test can correctly identify valid DNCTs from a set of candidate NCs.
	\begin{theorem}
		\label{DNCTtest_thm}
		Let the data be generated by a simple NC model, $\mathcal{G}$, and assume Assumptions~\ref{ass:1}, \ref{ass:2}, and \ref{ass:SNCMlinear}. The DNCT validation test will return TRUE for any DNCT in $\mathcal{G}$, and FALSE otherwise.
	\end{theorem}
	
	 We present a rationale for the theorem. A formal proof is provided in Appendix~\ref{appendix:technicallemma}.
	We make critical use of prior work  in proving correctness of the DNCT validation test \citep{sullivant2010trek,spirtes2013calculation}. In particular, we will heavily leverage two theorems in \citet{spirtes2013calculation} which we present below. In these theorems, $A$, $B$, $S_A$, and $S_B$ are sets of variables.
	
	\begin{theorem}{Extended Trek Separation Theorem 1 (Spirtes 2013).}\label{thm:PETS1}
		Suppose $\mathcal{G}$ is a directed graph containing $S_A$, $A$, $S_B$, and $B$, and $(S_A;S_B)$ t-separates $A$ and $B$ in $\mathcal{G}$. Then for all covariance matrices entailed by a fixed parameter structural equation model $S$ with path diagram $\mathcal{G}$ that is linear acyclic below the sets $S_A$ and $S_B$ for $A$ and $B$, $rank(\Sigma_{A,B}) \leq |S_A| + |S_B|$.
	\end{theorem}
	
	\begin{theorem}{Extended Trek Separation Theorem 2 (Spirtes 2013).}\label{thm:PETS2}
		For all directed graphs $\mathcal{G}$, if there does not exist a pair of sets $S_A$, $S_B$, such that $(S_A; S_B)$ t-separates $A$ and $B$ and $|S_A| + |S_B| \leq r$, then for any $S_A$, $S_B$ there is a fixed parameter structural equation model $S$ with path diagram $\mathcal{G}$ that is linear acyclic below the sets $(S_A; S_B)$ for $A$ and $B$ that entails $rank(\Sigma_{A,B}) > r$.
	\end{theorem}
	
	By making Assumption~\ref{ass:SNCMlinear}, we can directly apply the above theorems and use them in our proof of correctness. In particular, we statistically test whether certain subcovariance matrices are rank deficient, that is, if their rank is less than their dimension. For our application to the simple NC models, we will consider subcovariance matrices of the measured variables with dimension two. The presence of $U$ ensures that there is always one trek between any two sets of variables $A$ and $B$, and the rank of their subcovariance matrix will be at least one. If any other trek connects $A$ to $B$ and does not pass through $U$, it will force $S_A$ or $S_B$ to contain an additional variable, increasing their collective size to two. Since two is the dimension of the subcovariance matrix, it would not be rank deficient, and the determinant of the subcovariance matrix would not be 0. In this way, these tetrad tests check to see if there is any trek other than the ones through $U$ that would connect a member on one side of the tetrad to a member on the other side of the tetrad.

\section{The Find Negative Controls Algorithm}
\label{FNCalgo}	
	The Find Negative Controls (FindNC) algorithm, summarized in Algorithm~\ref{FNC pseudocode}, searches through a provided set of candidate NCs to identify triplets of candidate NCs that form a simple NC model along with the provided treatment $T$  and outcome $O$. It uses the DNCT validation test introduced in Section~\ref{sec:DNCTtest}, and performs a brute force search through the space of all candidate NC triplets. It outputs a collection of all the candidate NC triplets that passed the DNCT test, for a given $T$, $O$, and data set. 
	
	The FindNC algorithm takes the same input as the DNCT validation test, except that instead of investigating a set of three candidate NCs, the FindNC algorithm takes the set of all candidate NCs as input to identify DNCTs. The FindNC algorithm has one hyperparameter, \emph{alpha}, which is the threshold used for rejecting the null hypothesis in the six vanishing tetrad tests each time it is applied. In our implementation of the FindNC algorithm, the value for this hyperparameter is optional, and it will default to a value of $n^{-1}$, where $n$ is the provided data set's sample size. Numerical experiments have indicated that this is a reasonable heuristic for choosing {alpha}, with consistently well-balanced performance at a large range of sample sizes.

	The FindNC algorithm is a brute force search of the space of candidate NC triplets, using the DNCT validation test. Its correctness thus depends on the completeness of its search, and the correctness of the DNCT validation test. It is clear from its construction that FindNC checks all possible candidate NC triplets, so its search is complete. The correctness of the DNCT validation test was already proven (see Theorem \ref{DNCTtest_thm}). As such, FindNC is correct under the same conditions as the DNCT validation test.
	
	\begin{algorithm}
		\DontPrintSemicolon
		\KwData{$Data$ on a set of Candidate NCs (denoted as $CandidateNCs$), treatment $T$, and outcome $O$, the $VanTetrad$ function, the DNCT validation test introduced in Section~\ref{sec:DNCTtest} (denoted as $DNCTvalidation$), and the threshold parameter $alpha$}
		\KwResult{A collection of validated DNCTs}
		$Output \leftarrow \varnothing$\;
		\For{$X,Y,Z \in CandidateNCs$, and $X\neq Y\neq Z$}{
			\If{$DNCTvalidation$(Data,X,Y,Z,T,O,alpha,VanTetrad) returns TRUE}{
				$Output \leftarrow Output \cup \{X,Y,Z\}$
			}
		}
		\KwRet{Output}
		\caption{Find Negative Controls (FindNC)
		}
		\label{FNC pseudocode}
	\end{algorithm}

	This brute force search scales approximately $|V|^3$ for data with a set of variables $V$, since it scans through all possible triplets of variables. $|V|^3$ is computationally feasible for many real world domains where NCs might be used, unlike causal discovery which searches a space that scales super exponentially with $|V|$. This obviates the immediate need for more complex and efficient search procedures, and we leave the development of such procedures for future work.

	\section{Data-driven Automated Negative Control Estimation}
	\label{DNCcausalinference}
	
	In this section we combine the FindNC algorithm with
	NC estimation of the causal effect of treatment on outcome. 
	We first briefly review the literature on identification and estimation of the ATE with a double-NC pair of $Z$ and $W$. Then we present methods to aggregate information from multiple NC pairs. 
	\subsection{A Brief Review of Double-Negative Control Methods\label{subsec:doublenc}} 
\cite{miao2018identifying} proved that, under a completeness assumption to ensure that $Z$ is sufficiently informative about $U$, 
the ATE is identified by
\[
\Delta = \int_{-\infty}^{\infty}\left\{h(w,t=1)-h(w,t=0)\right\}f(w)dw
\]
where $h(w,t)$ is any solution to 
\begin{equation*}\label{eq:bridge}
	E[O\mid T=t,Z=z] = \int_{-\infty}^{\infty}h(w,t)f(w\mid z,t)dw.
\end{equation*}
The function $h(w,t)$ has been referred to as the outcome confounding bridge function \citep{miao2018confounding}.
The NC framework has been extended to proximal causal inference with an alternative identification via the treatment confounding bridge \citep{tchetgen2020introduction,cui2020semiparametric}. The outcome and treatment confounding bridge functions are analogous to the outcome regression and propensity score models, respectively, in the classical setting when all confounders are measured. In the following, we focus on the outcome confounding bridge function, as similar estimation strategies can be applied to the treatment confounding bridge function. 

In practice, one could specify a parametric model $h(W,T;\alpha)$ and jointly estimate $\theta=(\alpha,\Delta)$ by generalized method of moments (GMM) \citep{miao2018confounding,tchetgen2020introduction} with the following moment restrictions
\begin{equation}\label{eq:moment}
	g(O,T,Z,W;\theta) = \begin{bmatrix}
		q(1,Z,T)\{O-h(W,T;\alpha)\}\\
		\Delta - \left\{h(W,T=1)-h(W,T=0)\right\}
	\end{bmatrix},
\end{equation}
where $q$ is a vector function of the same dimension as $\alpha$. Specifying an identity weighting matrix, the GMM estimator for $\theta$ is 
\begin{equation}\label{eq:estimator}
	\widehat{\theta} = \arg\min_{\theta} \overline{g}^{\top}\overline{g},
\end{equation}
where $\overline{g}=n^{-1}\sum_{i=1}^n g(O_i,T_i,Z_i,W_i;\theta)$ is the average over a sample of $n$ observations. Asymptotic variance can also be estimated \citep{white1980heteroskedasticity,hansen1982large,10.2307/1913610}. 
The above results and methods are still valid with measured confounders. Specifically, to adjust for measured confounders, denoted as $X$, we solve for moment condition $E[g(O,T,Z,W,X;\theta)]=0$ where
\begin{equation}\label{eq:moment_x}
	g(O,T,Z,W,X;\theta) = \begin{bmatrix}
		q(1,Z,T,X)\{O-h(W,T,X;\alpha)\}\\
		\Delta - \left\{h(W,T=1,X)-h(W,T=0,X)\right\}
	\end{bmatrix}.
\end{equation}

In the special case where each child-parent family in the graph $\mathcal{G}$ represents a linear SEM, we have
\begin{equation*}\label{eq:h_linearSEM}
	h(W,T;\theta)=\alpha_0+\alpha_1 W+\Delta T,
\end{equation*}
where $\theta=(\alpha_0,\alpha_1,\Delta)^\top$, and correspondingly,
\begin{equation*}\label{eq:moment_linearSEM}
	g(O,T,Z,W;\theta) = (1,Z,T)^\top\{O-h(W,T;\theta)\}.
\end{equation*}
Therefore the moment condition $E[g(O,T,Z,W;\theta)]=0$ can be solved by two-stage least squares \citep{angrist1995two,wooldridge2010econometric,tchetgen2020introduction}.
In fact, under linear SEM, \cite{kuroki2014measurement} showed that the causal effect identified from a pair of disconnected NCs has a closed form solution which is given by
\begin{equation*}
    \begin{split}
	\Delta &=\frac{
		cov(T,O)cov(Z,W)-cov(Z,O)cov(T,W)
	}{
		cov(T,T)cov(Z,W)-cov(T,Z)cov(T,W)
	}\\
	&=
	\frac{
		cov(T,O)cov(Z,W)-cov(W,O)cov(T,Z)
	}{
		cov(T,T)cov(Z,W)-cov(T,Z)cov(T,W)
	},
\label{eq:ate_pearl}
\end{split}
\end{equation*}
% \begin{equation}
% 	\Delta =\frac{
% 		\sigma_{to}\sigma_{zw}-\sigma_{oz}\sigma_{tw}
% 	}{
% 		\sigma_{tt}\sigma_{zw}-\sigma_{tz}\sigma_{tw}
% 	}
% 	=
% 	\frac{
% 		\sigma_{to}\sigma_{zw}-\sigma_{ow}\sigma_{tz}
% 	}{
% 		\sigma_{tt}\sigma_{zw}-\sigma_{tz}\sigma_{tw}
% 	},
% \label{eq:ate_pearl}
% \end{equation}
%where $\sigma_{ab}=cov(X,Y)$ for a pair of random variables $X,Y$, and
where the second equality is due to the vanishing tetrad $cov(Z,O)cov(T,W)- cov(W,O)\newline  cov(T,Z)  =0$. 
%One can estimate the ATE by plugging the sample covariance into~\eqref{eq:ate_pearl}.

\subsection{Aggregate Multiple Double-Negative Controls to Estimate the ATE}\label{sec:aggregate}
The FindNC algorithm outputs one or more validated DNCTs. For one triplet, one can define six possible double-NC pairs each generating a distinct ATE estimate, while for multiple triplets, there could be overlapping NC pairs. In this section, we propose methods to aggregate all possible pairs of double-NCs to estimate the ATE assuming an outcome confounding bridge function $h(W,T,X;\alpha)$ defined in~\eqref{eq:moment_x}. 

It is important to note that $\alpha$ in the outcome confounding bridge function $h(W,T,X;\alpha)$ is the same regardless of what $Z$ is, and $\Delta$ is the same regardless of what the outcome confounding bridge function is. Therefore, for different NC pairs, $\alpha$ may be shared and $\Delta$ must be shared. Jointly estimating all parameters while acknowledging that certain parameters are shared by different NC pairs can improve efficiency.
To do so, one could stack all moment restrictions $g(O,T,Z,W;\theta)$ each corresponding to a double-NC pair, then estimate the unique parameters via GMM. 
We provide details about this approach in Appendix~\ref{appendix:ideal_est}. 

In practice, such a joint estimation method can be computationally challenging due to jointly estimating an excessive number of parameters. 
We thus propose two practical approaches to aggregate information from multiple DNCTs selected by the FindNC algorithm. The first is a majority vote method: we take the most frequently selected pair of NC variables among all triplets, then follow Section~\ref{subsec:doublenc} to estimate a single ATE. This majority vote method is convenient but does not fully utilize all available information. 

The second method is to estimate the ATE based on each double-NC pair, and then take a weighted average, where the weights are proportional to the frequency that each NC pair is selected. Asymptotic variance of the weighted average can be estimated by nonparametric bootstrapping. Alternatively, one could aggregate the moment restrictions (\ref{eq:moment}) from all double-NCs and then compute a sandwich variance estimator that takes into account the correlation between different ATE estimates obtained from different NC pairs. We summarize this method in Algorithm~\ref{ATE estimation}, referred to as the Aggregated Negative Control Estimation. 
	
	\begin{algorithm}[H]
		\DontPrintSemicolon
		\KwData{$Data$, DNCTs, treatment $T$, and outcome $O$}
		\KwResult{A point estimate $\widehat{\Delta}$ and confidence interval (CI) for the effect of $T$ on $O$ aggregated from all NC pairs in the DNCTs}
		$Estimates \leftarrow \varnothing$\;
		\For {$(NC1,NC2,NC3) \in DNCTs$}{
			\For {each pair $(NC,NC') \in (NC1,NC2,NC3)$}{
				Let $Z=NC$, $W=NC'$, solve for $\widehat{\Delta}$ from (\ref{eq:estimator})\;
				$Estimates \leftarrow Estimates \cup \{\widehat{\Delta}\}$\;
				Let $Z=NC'$, $W=NC$, solve for $\widehat{\Delta}$ from (\ref{eq:estimator})\;
				$Estimates \leftarrow Estimates \cup \{\widehat{\Delta}\}$\;
			}
		}
	Aggregate point estimates $\widehat{\Delta}\leftarrow$ weighted average of \textit{Estimates}\;
	Compute confidence interval via either nonparametric bootstrapping or GMM with aggregated moment restrictions detailed in Appendix~\ref{sec:variance}\;
		\KwRet{$\langle \widehat{\Delta}$, CI$\rangle$}\;
		\caption{Aggregated Negative Control Estimation}
		\label{ATE estimation}
	\end{algorithm}
		
\subsection{Data-driven Automated Negative Control Estimation}
Finally, we combine the FindNC algorithm (Algorithm~\ref{FNC pseudocode}) and the Aggregated Negative Control Estimation algorithm (Algorithm~\ref{ATE estimation}) into the DANCE algorithm, summarized in Algorithm~\ref{DANCE pseudocode}.
	DANCE first identifies subsets of a set of candidate NC variables that satisfy the assumptions of disconnected NCs, then applies the double-NC
	 method to each pair of these validated NC variables, and finally aggregates information from the outputs to produce an unbiased point estimate and confidence interval for the causal effect of interest in the presence of an unmeasured confounder. 
\begin{algorithm}
		\DontPrintSemicolon\KwData{$Data$ on a set of Candidate NCs (denoted as $CandidateNCs$), treatment $T$, and outcome $O$, the $VanTetrad$ function, the FindNC algorithm, the Aggregated Negative Control Estimation (AggregatedNCE) algorithm and the threshold parameter $alpha$}
\KwResult{A point estimate $\widehat{\Delta}$ and confidence interval (CI) for the effect of $T$ on $O$}
DNCTs $\leftarrow$ $FindNC(Data,CandidateNCs,T,O,VanTetrad,alpha$)\;
		$\langle\widehat{\Delta}, CI\rangle\leftarrow$ $AggregatedNCE(Data,DNCTs,T,O,\theta)$\;
\KwRet{$\langle\widehat{\Delta}, CI\rangle$}
		\caption{Data-driven Automated Negative Control Estimation (DANCE)}
		\label{DANCE pseudocode}
	\end{algorithm}

	\section{Simulation Study}
	\label{simulations}
	We perform simulation studies to evaluate the performance of the DANCE algorithm in detecting NCs and estimating the ATE.
	In these studies, data are generated based on linear SEMs under two graphical structures: a simple graph and a complex graph, as presented in Figure~\ref{fig:Graph1}, in which red arrows may lead to violation of the DNCT assumptions for certain NC candidate variables. We first simulate an unmeasured confounder, then all the other variables are generated based on linear SEMs, with edge strengths randomly sampled from uniform distributions. Details about the model parameters and data generating mechanisms are detailed in Appendix \ref{appendix:datagenerate}. 
	We consider various sample sizes ranging from 10 to 10,000. Simulation results are summarized over 200 replications.
	\begin{figure}[h]	
\centering
\begin{subfigure}[b]{0.4\textwidth}
\centering
\begin{tikzpicture}
\node[est] (A) at (0.5,-1.5) {T};
\node[est] (Y) at (3.5,-1.5) {O};
\node[shade] (U) at (2,0.5-0.3) {U};
\node[est] (Z1) at (-0,0.5-0.7) {Z1};
\node[est] (Z2) at (-0,0.5+1) {Z2};
\node[est] (Z3) at (4,0.5-0.7) {Z3};
\node[est] (Z4) at (4,0.5+1) {Z4};
\foreach \from/\to in {A/Y,U/A,U/Y,U/Z1,U/Z2,U/Z3,U/Z4}\path[->,pil,black] (\from) edgenode {} (\to);
\foreach \from/\to in {Z1/Z2}\path[->,pil,red] (\from) edgenode {} (\to);
\end{tikzpicture} 
\caption{Simple Graph}
\label{fig:SimGraph1}
\end{subfigure}
\begin{subfigure}[b]{0.4\textwidth}
\centering
\begin{tikzpicture}
\node[est] (A) at (0.5,-1.5) {T};
\node[est] (Y) at (3.5,-1.5) {O};
\node[shade] (U) at (2,0.5-0.8) {U};
\node[est] (Z1) at (-0,0.5-1.0) {Z1};
\node[est] (Z2) at (-0,0.5+0.3) {Z2};
\node[est] (Z3) at (0.9,0.5+0.7) {Z3};
\node[est] (Z4) at (2.0,0.5+1.4) {Z4};
\node[est] (Z5) at (3.1,0.5+0.7) {Z5};
\node[est] (Z6) at (4,0.5-1.0) {Z6};
\node[est] (Z7) at (4,0.5+0.3) {Z7};
\foreach \from/\to in {A/Y,U/A,U/Y,U/Z1,U/Z2,U/Z3,U/Z4,U/Z5,U/Z6,U/Z7}\path[->,pil,black] (\from) edgenode {} (\to);
\foreach \from/\to in {Z1/Z2,Z3/Z4,Z4/Z5,Z3/Z5,Z6/Z7}\path[->,pil,red] (\from) edgenode {} (\to);
\end{tikzpicture} 
\caption{Complex Graph}
\label{fig:SimGraph2}
\end{subfigure}
\caption{\label{fig:Graph1} The two different configurations used in the simulation study with a treatment $T$, an outcome $O$, an unmeasured confounder $U$, and multiple candidate NCs $Z_1,\dots,Z_7$. Due to the existence of red arrows, some candidate NC triplets are not DNCTs. }
\end{figure}
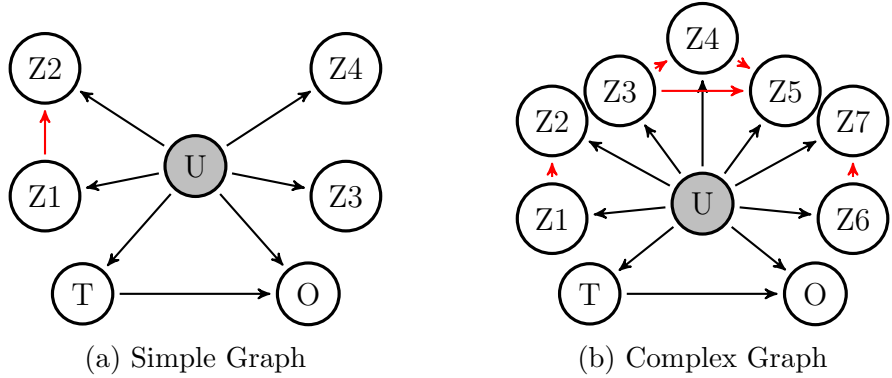
	
	With data generated from either simple or complex graphical structure, we conduct the following two evaluations.
	To assess the algorithm's ability to validate candidate NC variables, we plot ROC curves under varying thresholds used for rejecting the null hypothesis in the vanishing tetrad test. 
	To assess the algorithm's accuracy in estimating the causal effect, we compute proportion bias, variance, and coverage probability comparing DANCE algorithm with the following two methods: 
	\begin{itemize}\itemsep -3pt
	    \item[1.] \emph{Naive}, a naive regression method that ignores unmeasured confounding;
	    \item[2.] \emph{Random}, a random search method in which a pair of NCs is randomly selected from the candidate NCs to adjust for unmeasured confounding via the double-NC method.
	\end{itemize}

\begin{figure}[p]
	\centering
\begin{subfigure}[b]{0.45\textwidth}
\centering\includegraphics[width=7.5cm,height=6cm]{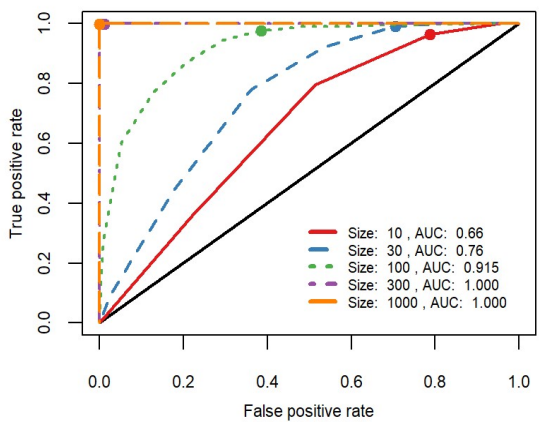} 
\caption{{\small ROC curve under simple graph}\label{fig:ROC_simp}} 
\end{subfigure}
\hspace{0.2in}
\begin{subfigure}[b]{0.45\textwidth}
\centering\includegraphics[width=7.5cm,height=6cm]{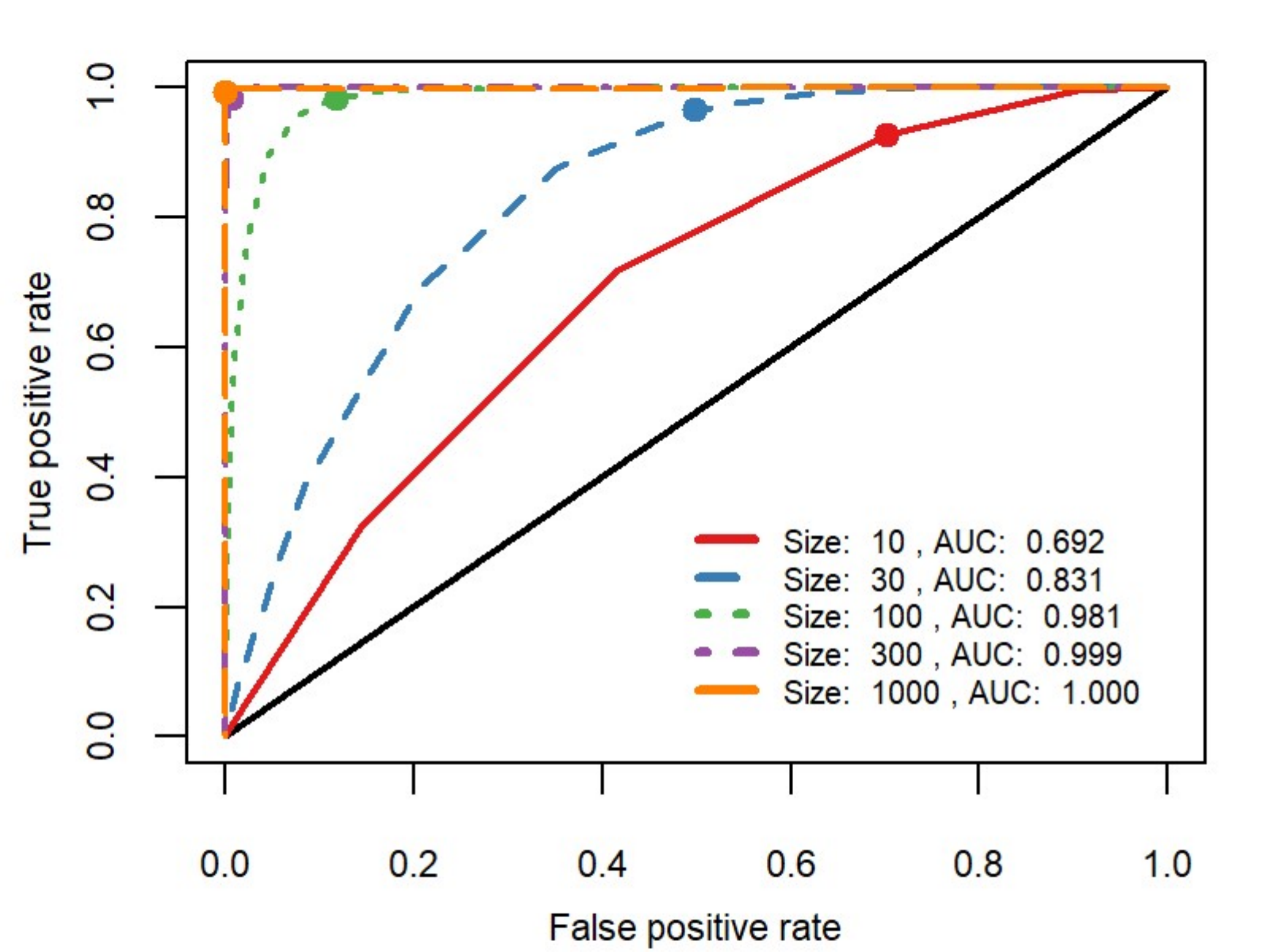} 
\caption{{\small ROC curve under complex graph}\label{fig:ROC_comp}} 
\end{subfigure}
\vskip\baselineskip
\begin{subfigure}[b]{0.45\textwidth}
\centering\includegraphics[width=7.5cm,height=6cm]{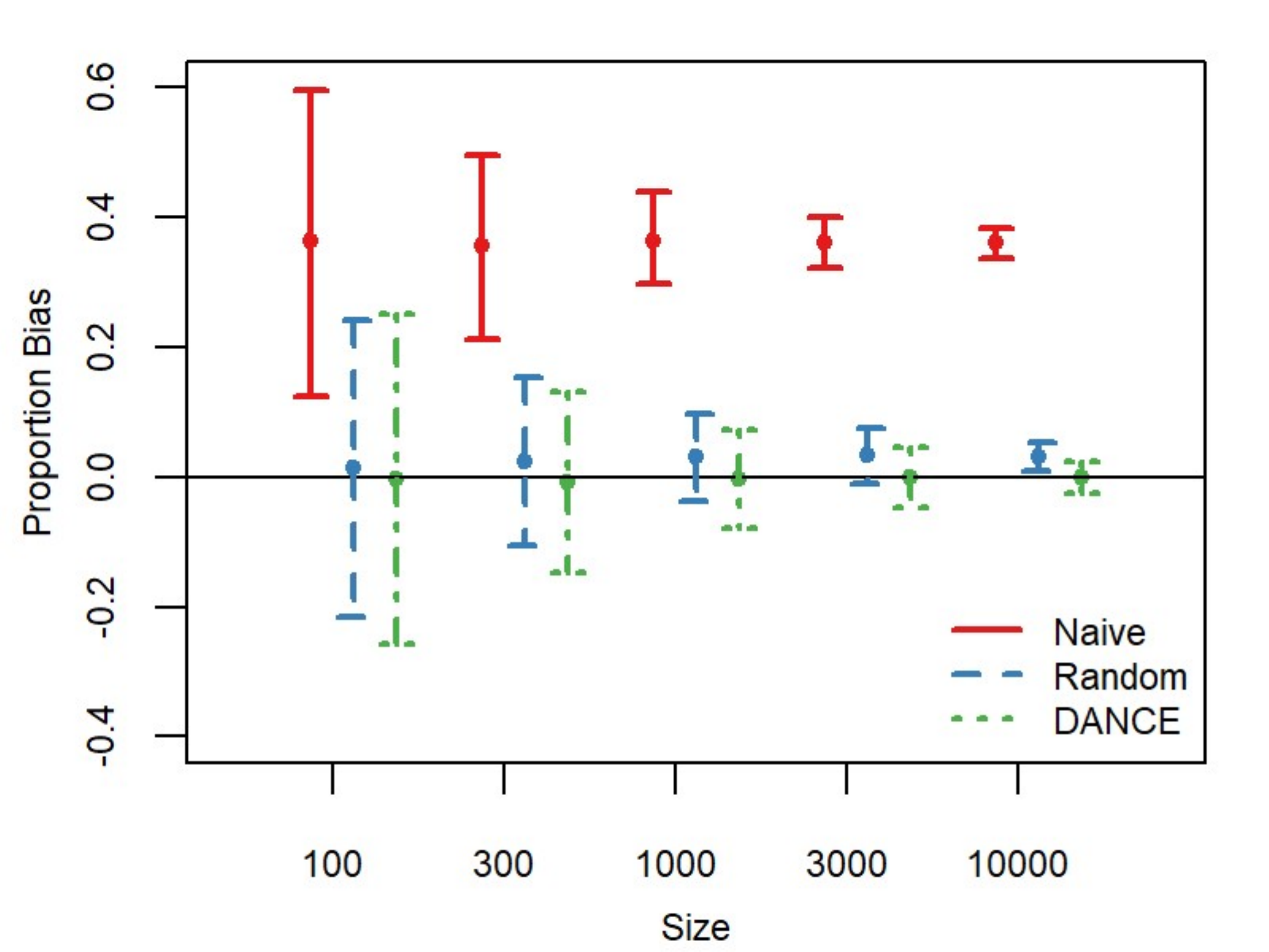}
\caption{{\small Proportion bias under simple graph}\label{fig:bias_simp}} 
\end{subfigure}
\hspace{0.2in}
\begin{subfigure}[b]{0.45\textwidth}
\centering\includegraphics[width=7.5cm,height=6cm]{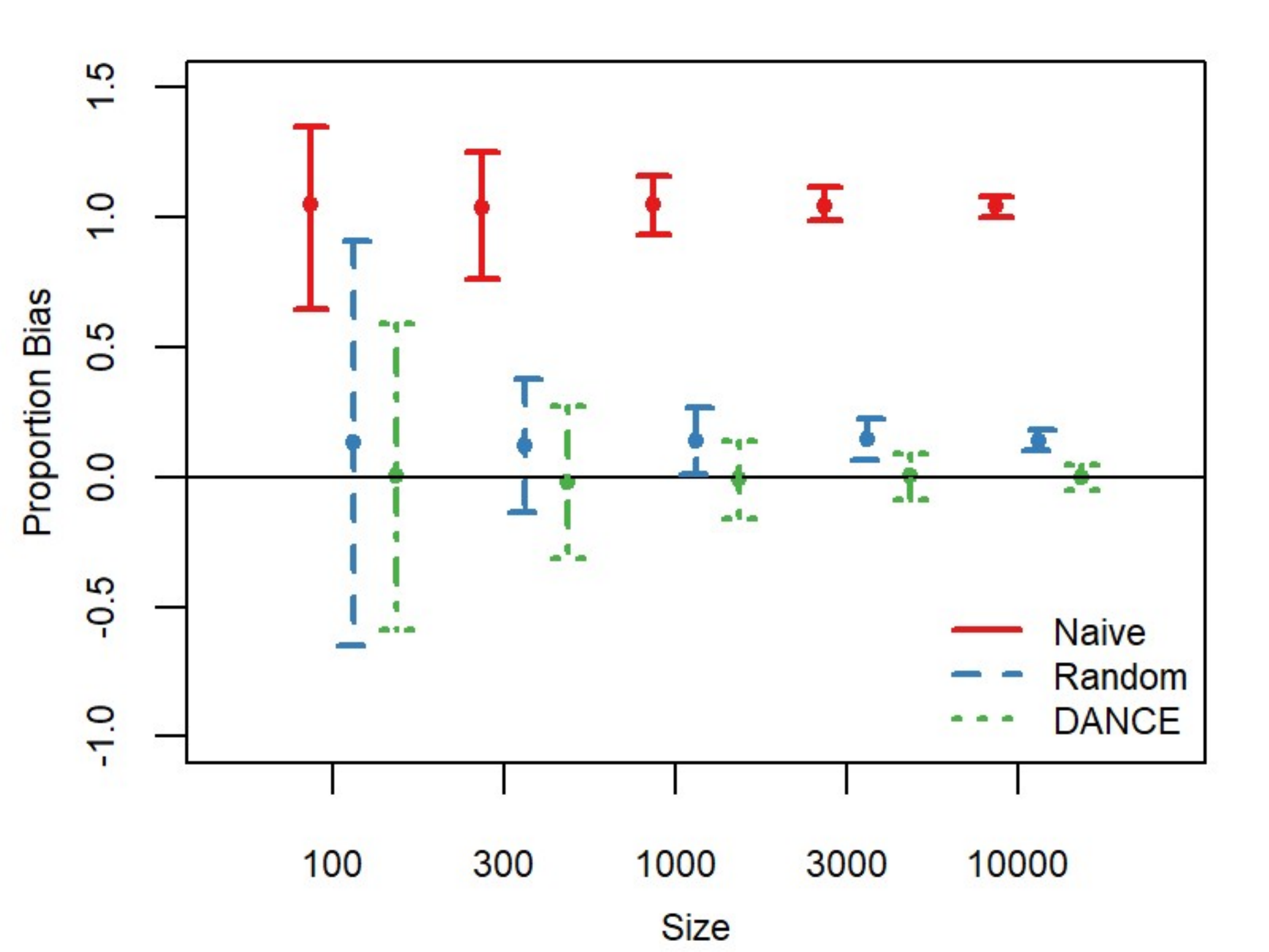}
\caption{{\small Proportion bias under complex graph}\label{fig:bias_comp}} 
\end{subfigure}
		\caption{Simulation results with data generated under simple and complex graphical structures. The first row presents ROC curves for NC validation, and the second row presents proportion of bias for each method in estimating the ATE. ``Naive" stands for naive method without adjusting for unmeasured confounding; ``Random" stands for randomly selecting an NC pair for estimation.} \label{fig:SimresGraph1}
	\end{figure}

	Figure~\ref{fig:SimresGraph1} presents the ROC curves for validation of NCs and the proportion of bias for estimation of the ATE. Under both simple (Figure~\ref{fig:ROC_simp}) and complex (Figure~\ref{fig:ROC_comp}) graphical structures, we observe improved ROC curve with near perfect separation between valid and invalid NCs when the sample size is greater than or equal to 300. 
	In terms of estimation bias, DANCE provides an unbiased estimate with improved precision as sample size increases. In contrast, both the Naive method and the Random method are substantially biased. We also note that the performance of the Random method depends on the proportion of valid NCs in the data. 
	When data are generated under the simple graphical structure as presented in Figure~\ref{fig:SimGraph1}, there are two DNCTs out of four candidate NC triplets: (Z1, Z3, Z4) and (Z2, Z3, Z4), while under the complex graphical structure as presented in Figure~\ref{fig:SimGraph2}, there are 12 DNCTs out of 35 candidate NC triplets. Therefore, the Random method has 50\% and 34\% chance of correctly selecting valid NCs under the simple and complex graph, respectively. Consequently, we observe larger bias in the Random method under complex graphical structure than  simple structure. 
	
	Table~\ref{tab:simu} presents additional information on the performance of ATE estimation based on the DANCE algorithm and the Random method. For both simple and complex graphs, the DANCE algorithm is able to provide much more accurate estimation of the ATE than the Random method, at the price of a slightly larger Monte Carlo standard error due to the search for NCs. In addition, the DANCE algorithm provides valid inference with coverage probability close to the nominal level of 95\%.

		\begin{table}[h]
	    \centering
	    \resizebox{\textwidth}{!}{\begin{tabular}{cccccccccccc}
\hline 
\multirow{2}{*}{Graph} & \multirow{2}{*}{Method} & \multicolumn{2}{c}{Bias} & \multicolumn{2}{c}{Proportion} & \multicolumn{2}{c}{Monte Carlo} & \multicolumn{2}{c}{Estimated} & \multicolumn{2}{c}{95\% CI}\tabularnewline
 &  & \multicolumn{2}{c}{($10^{-3}$)} & \multicolumn{2}{c}{Bias (\%)} & \multicolumn{2}{c}{SE ($10^{-3}$)} & \multicolumn{2}{c}{SE ($10^{-3}$)} & \multicolumn{2}{c}{Coverage}\tabularnewline
\hline 
\multicolumn{2}{c}{Sample size} & 1000 & 3000 & 1000 & 3000 & 1000 & 3000 & 1000 & 3000 & 1000 & 3000\tabularnewline
\hline 
\hline 
\multirow{2}{*}{Simple} & Random & 19.08 & 20.61 & 3.05 & 3.29 & 21.38 & 13.55 & 21.73 & 13.81 & 0.86 & 0.70\tabularnewline
 & DANCE & -1.82 & 0.06 & -0.29 & 0.01 & 23.88 & 14.67 & 22.93 & 14.90 & 0.92 & 0.96\tabularnewline
 &  &  &  &  &  &  &  &  &  &  & \tabularnewline
\multirow{2}{*}{Complex} & Random & 86.55 & 91.52 & 13.82 & 14.61 & 40.59 & 25.04 & 40.14 & 25.10 & 0.39 & 0.05\tabularnewline
 & DANCE & -5.45 & 1.11 & -0.87 & 0.18 & 47.46 & 28.06 & 45.21 & 28.66 & 0.95 & 0.94\tabularnewline
\hline 
\end{tabular}}
	    \caption{Operation Characteristics of the estimator based on randomly selected negative controls (the Random method) and the DANCE estimator for the simple and complex graphs.\label{tab:simu}}
	    \label{tab:my_label}
	\end{table}
	We further investigate the performance of our proposed method under stronger edge strength in the linear SEMs generating the data. In addition,  to assess the sensitivity of the vanishing tetrad test to violations of Gaussian assumptions, we  perform simulations in which all variables follow Bernoulli distributions rather than Gaussian distributions. The simulation results are presented in Appendix \ref{appendix:simul}. We observe similar results under these more challenging settings.

	\section{Demonstration on Real World Data}
	\label{demonstrations}
	We illustrate our proposed methods with an application to the Study to Understand Prognoses and Preferences for Outcomes and Risks of Treatments (SUPPORT) to evaluate the effectiveness of right heart catheterization (RHC) procedure among seriously ill hospitalized adults admitted to the intensive care unit (ICU) \citep{connors1996effectiveness}.
	Many physicians believed that measurements from the RHC procedure can guide therapy  and lead to better outcomes for critically ill patients. Due to its popularity and physicians' strong belief, conducting a clinical trial was unethical. In the absence of an RCT, the SUPPORT team conducted an observational study to evaluate the effectiveness of the RHC procedure. Out of 5,735 critically-ill patients who were considered for the RHC procedure on their admission to an ICU. RHC was performed in 2,184 patients, and the remaining 3,551 patients were managed without RHC. The study found that, contrary to expectation, RHC was associated with increased mortality.
	Due to this controversial result, this data set has been further analyzed by many researchers   \citep{lin1998assessing,tan2006distributional,li2018balancing,mao2020flexible,tchetgen2020introduction}. 
	
	A particular concern has been the potential of hidden bias due to confounding by indication. In this application, we aim to find and apply NC variables to estimate the  causal effect of RHC on 30-day survival defined as the number of days between admission and death or censoring at 30 days, while accounting for potential unmeasured confounding. 
	The SUPPORT study measured an extensive set of 72 covariates including demographics, comorbidity, vital signs, physiological status, and functional status.	
	We applied our DANCE algorithm to all measured covariates to find valid NCs. Our algorithm identified 43 DNCTs which resulted in 164 unique double-NC pairs out of a total of 258 pairs. The most frequently selected pair of NC variables were two comorbidity variables: one is \texttt{dementhx} which stands for dementia, stroke or cerebral infarct, Parkinson’s disease, and the other is \texttt{gibledhx} which stands for upper gastrointestinal (GI) bleeding. We applied the double-NC method introduced in Section~\ref{subsec:doublenc} using this pair while adjusting for the 42 measured covariates not selected into the DNCTs. 
	Direct adjustment of the measured covariates would lead to a large number of parameters to be estimated for each outcome confounding bridge. With 164 unique double-NC pairs, the total number of parameters to be estimated is extremely large. Therefore, instead of directly adjusting for measured covariates, we adjusted for cubic spline basis of the propensity score obtained by regressing the treatment on the 42 measured covariates via logistic regression.
	We also applied the method introduced in Section~\ref{sec:aggregate} to aggregate the estimated RHC effects based on all selected pairs. 
	We further implemented a simple linear regression that regresses the outcome on all measured covariates without any attempt to control for unmeasured confounding.
	
	Table~\ref{tbl:application} presents the estimated effect of RHC on days alive since admission to ICU up to 30 days. Similar to the simple linear regression, our methods estimated that RHC has a negative effect on 30-day survival among adults admitted to the ICU. However, different from the regression adjustment approach, our methods provided a 95\% confidence interval that covers zero, which is evidence of the uncertainty due to potential unmeasured confounding. An interesting observation is that we obtained very similar results when flipping variables (\texttt{dementhx} and \texttt{gibledhx}) allocated to $W$ and $Z$. As discussed in \cite{tchetgen2020introduction}, under the causal graph in Figure~\ref{fig:DAG}, the role of $W$ and $Z$ are equivalent, hence causal inference would remain invariant to the choice of $W$ and $Z$ in the double-NC method. Therefore the similarity in our results is strong evidence that the identified pair of NCs satisfy the causal graph in Figure~\ref{fig:DAG}. In addition, despite the relatively wide confidence intervals from the majority vote method, we were able to provide more precise inference with a narrow confidence interval by aggregating information from all selected NCs.

	\begin{table}[!h]
		\centering
\begin{tabular}{c|c|c}
\hline 
\multicolumn{2}{c|}{Method} & RHC effect on days alive (95\% CI)\tabularnewline
\hline 
\hline 
\multicolumn{2}{c|}{Regression ignoring unmeasured confounding} & -1.29 (-1.83, -0.75)\tabularnewline
\hline
Most frequent & W=\texttt{dementhx}$^\dag$, Z=\texttt{gibledhx}$^\ddag$ & -2.98 (-14.96, 8.99) \tabularnewline
\cline{2-3} \cline{3-3} 
NC pair & W=\texttt{gibledhx}, Z=\texttt{dementhx} & -2.80 (-13.42, 7.83)\tabularnewline
\hline 
\multicolumn{2}{c|}{Aggregate over all NC pairs} & -0.19 (-1.50,  1.12)\tabularnewline
\hline 
\end{tabular}
		\\
		$^\dag$  Dementia, stroke or cerebral infarct, Parkinson’s disease. $^{\ddag}$ Upper GI bleeding.
		\caption{Results from application to the right heart catherization study. 
			}
		\label{tbl:application}
	\end{table}

	\section{Discussion}
	\label{Discussion}
	
	In this paper we introduced a new test for validating candidate NC variables (i.e., the DNCT validation test), proved its correctness, implemented a correct search procedure using it (i.e., the FindNC algorithm), evaluated the search procedure's performance in simulations, and combined this search procedure with a negative control estimation procedure to produce the DANCE method for causal inference, which we evaluated with simulations and demonstrated on a real world data set. DANCE allows for causal inference in the presence of an unmeasured confounder, but does not assume that the user can identify precisely which variables in the data set meet the strict requirements of NCs. It instead asks only for a collection of candidate NC variables. If no subset of variables passes the DNCT validation test, then DANCE does not produce an estimate, but reports that no variables meet the condition. As such, DANCE does not assume that the provided set of variables includes valid NCs, either. 
	
	The methods we provide here still have some limitations. The validation method we provide will throw out some NCs that satisfy a DAG different than Figure~\ref{fig:DAG}. Our correctness proof for the NC validation test assumes continuous variables with linear relationships, while many real world scenarios involve binary or categorical variables. DANCE may still work for binary variables, as shown in our simulation study in Appendix~\ref{appendix:simul}, but this is currently lacking a formal proof. Our correctness proof currently also assumes that the data are generated by a large but still limited set of structures. It is possible to test more general assumptions and more complex structures, and to rule out candidate NCs that do not meet these structural requirements using methods from the causal discovery literature, but at present there is no proof for the correctness of a more complex method.
	
	Future work can improve the computational efficiency of the validation methods presented here, relax some of the assumptions or limitations, or apply similar techniques to different problems. More complex structural scenarios should also be considered, such as when there is more than one unmeasured confounder.

	\clearpage\appendix
	\renewcommand\theassumption{\ref{appendix:thm2}.\arabic{assumption}}
	\begin{center}
		{\bf \Large Supplementary Material for ``Data-driven Automated Negative Control Estimation (DANCE): Search for, Validation of, and Causal Inference with Negative Controls"}
	\end{center}
	
\def\a{A}
\def\b{B}
\def\c{C}

\section{Technical Lemmas and Proof of Theorem~\ref{DNCTtest_thm}}
\label{appendix:technicallemma}
We first present two technical lemmas that are essential to proving  Theorem~\ref{DNCTtest_thm}.
\begin{lemma}
		\label{LemmaAll}
		Let the data be generated by a simple NC model, $\mathcal{G}$.  Under Assumptions~\ref{ass:1} and~\ref{ass:SNCMlinear}, the DNCT validation test will return TRUE for any DNCT in $\mathcal{G}$.
	\end{lemma}
	\begin{proof}
		Let $\{A,B,C\}$ be an arbitrary DNCT in $\mathcal{G}$. By construction, our DNCT validation test will return TRUE only for candidate negative control triplets in the simple NC model $\mathcal{G}$ for which all 3 vanishing tetrad tests consisting of the members of the triplet and $T$, and all 3 vanishing tetrad tests consisting of the members of the triplet and $O$, do not reject the corresponding null hypothesis that the determinant of subcovariance matrix is 0. We proceed by showing that the null hypotheses investigated by these 6 tests will not be rejected.
		
		Since $\{A,B,C\}$ is an arbitrary DNCT in $\mathcal{G}$, by definition of DNCT $U$ must be on all treks from each of $\{A,B,C\}$ to every variable in $\{A,B,C\}$. By the construction of $\mathcal{G}$ as an simple NC model, $U$ is also on all treks from each member of $\{A,B,C\}$ to $T$ and to $O$. From this point, we can set ourselves up to apply Theorem~\ref{thm:PETS1}. By the construction of the simple NC model, $U$ is not only on all of the aforementioned treks, but it is also the source of those treks. As such, without loss of generality let $U \in S_A$ and let $S_B = \varnothing$ (as compared to the converse). Then for all partitions of $\{A,B,C,T\}$ and $\{A,B,C,O\}$ into sets $S_1$ and $S_2$ with $|S_1|=|S_2|$, $(S_A;S_B)$ t-separates $S_1$ and $S_2$ in $\mathcal{G}$, since $U$ is the source of all treks among every pair of variables in $\{A,B,C,T\}$, and the same for $\{A,B,C,O\}$. Finally, with this setup and Assumptions~\ref{ass:1} and~\ref{ass:SNCMlinear}, we can apply Theorem~\ref{thm:PETS1}, implying that $rank(\Sigma_{S_1,S_2}) \leq |S_A| + |S_B|$. Because $|S_A| + |S_B| = 1$, we have that $\Sigma_{S_1,S_2}$, a 2 by 2 matrix, is rank deficient, and has a determinant of 0. As such, $VanTetrad(S_1,S_2)$ will not reject the null hypothesis for all such sets $S_1$ and $S_2$.
	\end{proof}
	
	\begin{lemma}
		\label{LemmaOnly}
		Let the data be generated by a simple NC model, $\mathcal{G}$. Under Assumptions~\ref{ass:1}---\ref{ass:SNCMlinear}, the DNCT validation test will return FALSE for any set of three candidate negative controls that is not a DNCT in $\mathcal{G}$.
	\end{lemma}
	\begin{proof}
		Let $\{A,B,C\}$ be an arbitrary triplet of distinct candidate negative controls that is not a DNCT in $\mathcal{G}$. It will suffice to show that for some partition of $\{A,B,C,T\}$ (or $\{A,B,C,O\}$) into sets $S_1$ and $S_2$ with $|S_1|=|S_2|$, $VanTetrad(S_1,S_2)$ will reject the null hypothesis that $det(\Sigma_{S_1,S_2})=0$.
		
		Assumptions~\ref{ass:2} and~\ref{ass:SNCMlinear} together imply that if there is no $S_A$ and $S_B$ that t-separate $S_1$ from $S_2$ such that $|S_A|+|S_B| \leq 1$, then $rank(\Sigma_{S_1,S_2})>1$, and thus $det(\Sigma_{S_1,S_2})\neq 0$. In the large sample limit, $VanTetrad(S_1,S_2)$ would thus reject the null hypothesis that $det(\Sigma_{S_1,S_2})=0$. Therefore, it suffices to show that there is no $S_A$ and $S_B$ that t-separate $S_1$ from $S_2$ such that $|S_A|+|S_B| \leq 1$.
		
		Since $\{A,B,C\}$ is not a DNCT, this implies that there is at least one trek from one member of $\{A,B,C\}$ to another that does not pass through $U$. Without loss of generality, assume this trek goes from $A$ to $B$. Let $S_1=\{A,C\}$ and $S_2 = \{B,T\}$. By construction of the simple NC model, $U$ is the source of at least one trek from $A$ to $B$, with no other variables along that trek.
		
		Let sets $S_A$ and $S_B$ t-separate $S_1$ from $S_2$ with $S_A$ on the $S_1$ side of the t-separating set and $S_B$ on the $S_2$ side of the t-separating set. By excluded middle, either (I) $S_A$ includes $A$ or $S_B$ includes $B$, or (II) $S_A$ does not include $A$ and $S_B$ does not include $B$. We proceed by disjunctive elimination, showing that assuming either (I) or (II) entails that $|S_A|+|S_B| \geq 2$.
		
		First, assume (I). Further, without loss of generality assume that $A\in S_A$ (as compared to $B\in S_B$). This blocks both treks from $A$ to $B$. By the construction of $\mathcal{G}$ as an simple NC model, there is also a trek from $C$ to $T$ that passes only through $U$, so either $C$, $T$, or $U$ need to be included in either $S_A$ or $S_B$. Since $A$ is assumed to be in $S_A$, this entails that $|S_A|+|S_B| \geq 2$.
		
		Second, assume (II). Since $S_A$ and $S_B$ are assumed to t-separate $S_1$ from $S_2$, but $S_A$ does not include $A$, $S_B$ does not include $B$, and the extra trek between $A$ and $B$ does not include $U$, there must be another variable, $Y \not\in \{A,B,U\}$, that blocks the extra trek between $A$ and $B$ when placed into either $S_A$ or $S_B$. However $Y$ does not block the trek from $A$ to $B$ that passes only through $U$, so an additional variable ($U$) must be added to $S_A$ or $S_B$ in order to block that trek as well, thus $|S_A|+|S_B| \geq 2$.
		
		Therefore, by disjunctive elimination, $|S_A|+|S_B| \geq 2$. This implies that there is no $S_A$ and $S_B$ that t-separate $S_1$ from $S_2$ such that $|S_A|+|S_B| \leq 1$, and so the DNCT validation test will return FALSE for $\{A,B,C\}$.
	\end{proof}
	
  \paragraph{Theorem 3.} Let the data be generated by a simple NC model, $\mathcal{G}$, and assume Assumptions~\ref{ass:1}, \ref{ass:2}, and \ref{ass:SNCMlinear}. The DNCT validation test will return TRUE for any DNCT in $\mathcal{G}$, and FALSE otherwise.
	\begin{proof}
		Theorem~\ref{DNCTtest_thm} follows directly from Lemmas \ref{LemmaAll} and \ref{LemmaOnly}.
	\end{proof}

	\section{Joint Estimation and Inference via Aggregated Moment Restrictions}\label{appendix:ideal_est}
	
\subsection{Aggregated Estimation With One or More DNCTs}
We illustrate our method to aggregate multiple double-negative control pairs across one or more validated DNCTs taking one triplet as an example. 
For one triplet, $\{\a, \b, \c\}$, one can define six possible double-negative control pairs each generating a distinct ATE, which are  $\{W = \a, Z = \b\}, \{W = \b, Z = \a\}, \{W = \a, Z = \c\}, \{W = \c, Z = \a\}, \{W = \b, Z = \c\}, \{W = \c, Z = \b\}\}$. 
Correspondingly, we have six moment restrictions and we stack them into the following joint moment restrictions:
\begin{align*}
    g(&O,T,\a,\b,\c,X; \alpha_{\a},\alpha_{\b},\alpha_{\c},\Delta)\\
    =
    &\begin{bmatrix}
    q(1,Z=\b,T,X)\{O-h(W=\a,T,X;\alpha_{\a})\\
    q(1,Z=\c,T,X)\{O-h(W=\a,T,X;\alpha_{\a})\\
    q(1,Z=\a,T,X)\{O-h(W=\b,T,X;\alpha_{\b})\\
    q(1,Z=\c,T,X)\{O-h(W=\b,T,X;\alpha_{\b})\\
    q(1,Z=\a,T,X)\{O-h(W=\c,T,X;\alpha_{\c})\\
    q(1,Z=\b,T,X)\{O-h(W=\c,T,X;\alpha_{\c})\\
    \Delta - \left\{h(W=\a,T=1,X;\alpha_{\a})-h(W=\a,T=0,X;\alpha_{\a})\right\}\\
    \Delta - \left\{h(W=\b,T=1,X;\alpha_{\b})-h(W=\b,T=0,X;\alpha_{\b})\right\}\\
    \Delta - \left\{h(W=\c,T=1,X;\alpha_{\c})-h(W=\c,T=0,X;\alpha_{\c})\right\}\\
    \end{bmatrix},
\end{align*}
where $\alpha_i$, $i\in\{\a,\b,\c\}$ is the same regardless of what $Z$ is, and $\Delta$ is the same regardless of what the outcome confounding bridge function is.
For example, $\a$ is an NCO in two negative control pairs: $\{W = \a, Z = \b\}$ and $\{W = \a, Z = \c\}$. Consequently, there are two moment functions that share the same parameter $\alpha_{\a}$ in the outcome confounding bridge function. With three disconnected negative controls, there are three unique outcome confounding bridge functions $h(W=i,T=1,X;\alpha_i)$, $i\in\{\a,\b,\c\}$ .
In addition, all three confounding bridge functions lead to the same ATE $\Delta$, that is, $\Delta=E[h(W=i,T=1,X;\alpha_i)-h(W=i,T=0,X;\alpha_i)]$, for all $i\in\{\a,\b,\c\}$.
With the aggregated moment restrictions, we can estimate and make statistical inference on $(\alpha_{\a},\alpha_{\b},\alpha_{\c},\Delta)$ simultaneously using GMM. 

With multiple DNCTs, there would be more NCOs each corresponding to an outcome confounding bridge function. As such, the number of $\alpha$ parameters as well as the dimension of parameters will grow rapidly, making it computationally challenging to estimate all parameters jointly. The disadvantage of the joint estimation method in terms of lack of computational efficiency may offset the advantage in terms of improved statistical efficiency.

\subsection{Variance Estimation \label{sec:variance}}
For one triplet, one can define six possible negative control pairs each generating a distinct ATE estimate, while for multiple triplets, there could be overlapping negative control pairs. 

For two DCNTs, if there are overlapping negative control pairs selected,  then we will have different frequency for different negative control pairs. For example, suppose we identified two DCNTs, \{NC1,NC2,NC3\} and \{NC1,NC2,NC4\}, in which NC1 and NC2 appears in both triplets. Then there  are ten rather than twelve unique negative control pairs, and the frequency of negative control pairs \{$W$ = NC1, $Z$ = NC2\} and \{$W$ = NC2, $Z$ = NC1\} is two while the frequency of the rest of the negative control pairs is one.  We take an weighted average of the corresponding ten ATE estimates where the weights are proportional to the frequency. Because the aggregated ATE is a linear combination of the unique ATE estimates from individual negative control pairs, it becomes clear that if we can compute the variance-covariance matrix of the ATEs estimated from the unique negative control pairs, then the variance of the aggregated ATE can be computed. We detail our inference method for the aggregated ATE below.

We first introduce notation. Suppose there are $K$ unique negative control pairs denoted by $\{Z^{k},W^{k}:~k=1,\dots,K\}$. We observe a sample of $n$ observations and for the $k$-th negative control pair we use $\{T_i, O_i, Z_i^{k},W_i^{k}:~i=1,\dots,n\}$ to obtain the $k$-th ATE estimate. Specifically, let $\widehat{\theta}^{k}=(\widehat{\alpha}^{k},\widehat{\Delta}^{k})= \arg\min_{\theta} \overline{g^{k}}^{\top}\overline{g^{k}}$ where 
\[
\overline{g^{k}}  = \overline{g^{k}}(O_i,T_i,W_i^{k},Z_i^{k};\theta)  = \frac{1}{n}\sum_{i=1}^n g(O_i,T_i,W_i^{k},Z_i^{k};\theta)
\]
and $g()$ is given in eq. \eqref{eq:moment}. Then we define the final aggregated ATE as a weighted average of the ATE estimates from the negative control pairs
\[
\widehat{\Delta} = \sum_{k=1}^{K} w_k \widehat{\Delta}^{k},
\]
where the weight $w_k$ is proportional to the frequency of the $k$-th negative control pair. 
%%%
Let $\widehat{\theta} = \left\{(\widehat{\theta}^{1})^\top,\dots,(\widehat{\theta}^{K})^\top\right\}^\top$ and
\begin{equation*}
	\begin{split}
		G_n(\widehat{\theta}) &= \left\{\overline{g^{1}}(O_i,T_i,W_i^{1},Z_i^{1};\widehat{\theta}^{1})^\top, \dots, 
		\overline{g^{K}}(O_i,T_i,W_i^{K},Z_i^{K};\widehat{\theta}^{K})^\top\right\}^\top\\
		A_n(\widehat{\theta}) &= \frac{\partial}{\partial \theta} G_n(\widehat{\theta})	\\
		B_n(\widehat{\theta}) &= n G_n(\widehat{\theta}) G_n(\widehat{\theta}) ^\top,
	\end{split}
\end{equation*}
then we have the following empirical sandwich estimator of the variance-covariance matrix of $\widehat{\theta}$
\[
V(\widehat{\theta}) = A_n(\widehat{\theta}) ^{-1}B_n(\widehat{\theta})  A_n(\widehat{\theta}) ^{-1}.
\]
Note that the aggregated ATE is a linear combination of $\widehat{\theta}$ given by
\[
\widehat{\Delta} =\omega^\top \theta,
\]
where $\omega=(\overrightarrow{0}_1,w_1,\dots,\overrightarrow{0}_K,w_K)^\top$ and
 $\overrightarrow{0}_k$ is a zero vector of the same dimension as $\alpha^{k}$, $k=1,\dots,K$. Therefore, the variance of $\widehat{\Delta}$ is
\[
V(\widehat{\Delta}) = \omega^\top V(\widehat{\theta}) \omega.
\]

\section{Data Generating Mechanisms in Simulation Studies}
\label{appendix:datagenerate}
The data generation procedure of the simulation studies is presented below.

In each scenario, data are generated
based on linear SEMs under a given graphical structure (Figure~\ref{fig:Graph1}). 
Specifically, in the Gaussian graphical model scenario, we first simulate an unmeasured confounder $U$ following a $Normal(0, 2)$ distribution. Then all the other variables are generated based on linear SEMs with exogenous error terms following a $Normal(0, 1)$ distribution. 
The coefficients in the SEMs were randomly simulated from uniform distributions with certain parameters as detailed in Table~\ref{tab:datgen}, and then fixed over all iterations. We generate edges that potentially lead to violation of the disconnected negative control assumptions, i.e., edges between negative controls highlighted in red color in Figure~\ref{fig:Graph1}, such as the $Z_1 \rightarrow Z_2$ in Figures~\ref{fig:SimGraph1} and \ref{fig:SimGraph2}.
We make the strength of such edges stronger by generating the coefficient of negative controls from the uniform distribution with larger parameters than the other coefficients (coefficients of $U$ and $T$) as detailed in Table~\ref{tab:datgen}. For example, under the weak edge strength scenario, coefficients of  $U$ and $T$ follow a Uniform distribution between 0.3 and 0.7, and the coefficients of negative controls follow a Uniform distribution between 1.0 and 2.0.

In the binary graphical models scenario, all variables are generated from Bernoulli distributions. We first simulate an unmeasured confounder from a Bernoulli distribution with success probability 0.5. Then we generate all other nodes following Figure~\ref{fig:Graph1} from Bernoulli distributions with success probability being the sigmoid function of a linear structural equation. 
Coefficients of all variables (which are on the log odds ratio scale) are generated from a Uniform distribution between 1.0 and 2.0, and the intercept is fixed at  -1.0.

\begin{table}[!htp]
\centering
\begin{tabular}{|c|c|c|c|c|}
\hline 
\multirow{2}{*}{Distribution} & Edge & Coefficient & Coefficient & Unmeasured\tabularnewline
 & Strength & of $U$ and $T$ & of NCs & Confounder\tabularnewline
\hline 
\multirow{2}{*}{Gaussian} & weak & Unif(0.3, 0.7) & Unif(1.0, 2.0) & Normal(0, 2)\tabularnewline
 & strong & Unif(0.6, 1.0) & Unif(2.0, 4.0) & Normal(0, 2)\tabularnewline
\hline 
Binary &  & Unif(1.0, 2.0) & Unif(1.0, 2.0) & Bernoulli(0.5)\tabularnewline
\hline 
\end{tabular}
\caption{\label{tab:datgen}Parameters in data generating SEMs under  Gaussian and binary graphical models.}
\end{table}

\section{Additional Simulation Results \label{appendix:simul}}
To test robustness of the DANCE algorithm, we performed  simulation studies under two additional scenarios: 
\begin{itemize}%\itemsep -4pt
\item Stronger edges in the Gaussian graphical models;
\item All variables are binary following Bernoulli distributions. 
\end{itemize}
Figure~\ref{fig:SimresGraph2} shows the ROC curves for validation of negative controls and the proportion of bias for estimation of the ATE when the data were generated under stronger edge strength than Figure~\ref{fig:SimresGraph1}. We observe even more improved ROC curve with near perfect separation between valid and invalid negative controls when the sample size is greater than or equal to 100. In addition, DANCE provides an unbiased estimate of ATE but the Naive and Random method showed statistically significant bias even with large sample size for both graphs.
Table~\ref{tab:simu2} presents additional information on the performance of ATE estimation based on the DANCE algorithm and the Random method where the data were generated from stronger edge strength. Similar to what we observe in Table~\ref{tab:simu},  DANCE algorithm has near 95\% coverage probability  while the Random method has low coverage particularly with the large sample and under the complex graph structure.

\begin{figure}[!htp]
	\centering
\begin{subfigure}[b]{0.45\textwidth}
\centering\includegraphics[width=7.5cm,height=6cm]{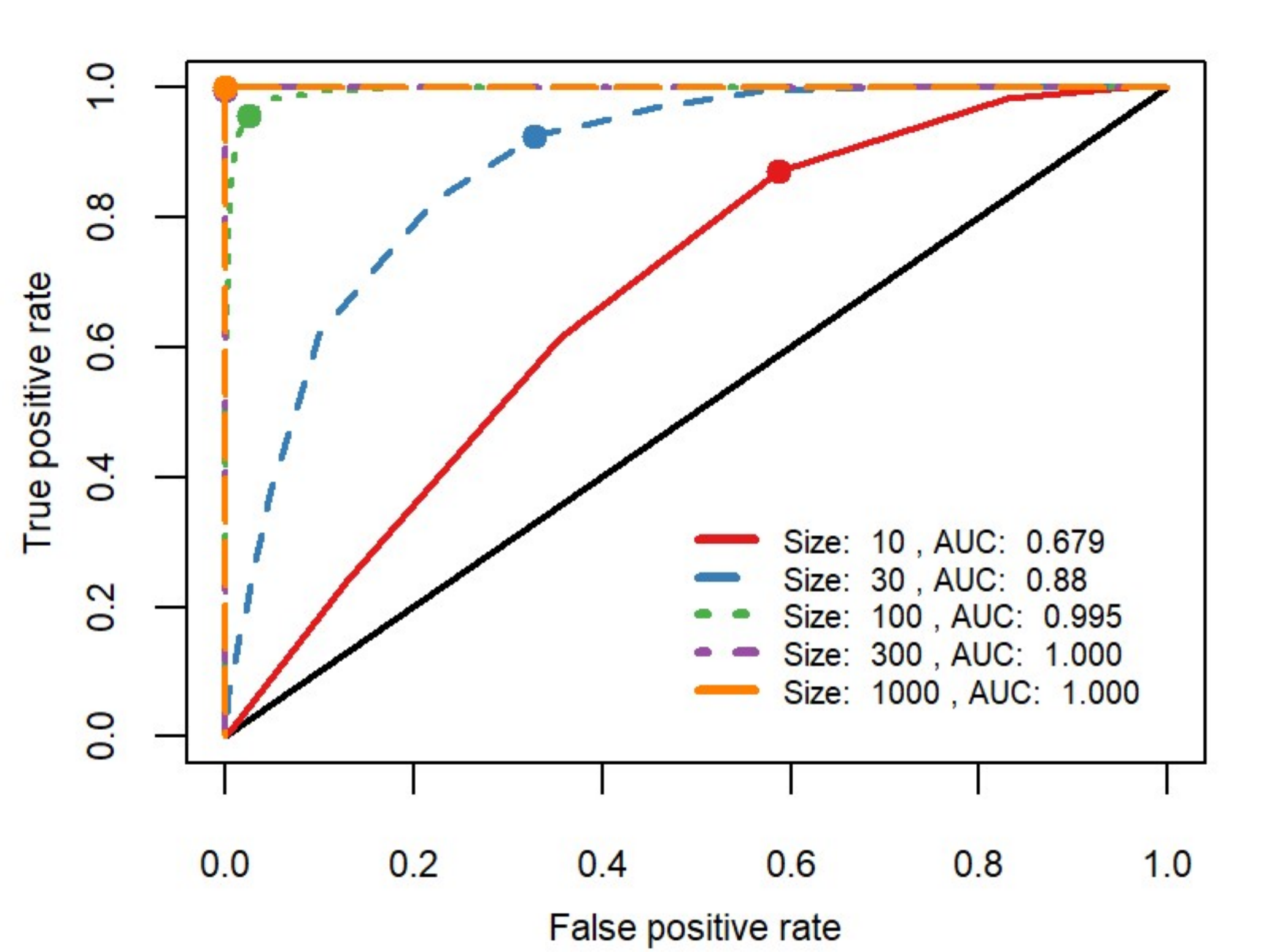}
\caption{{\small ROC curve under simple graph}\label{fig:ROC_simp_str}} 
\end{subfigure}
\hspace{0.2in}
\begin{subfigure}[b]{0.45\textwidth}
\centering\includegraphics[width=7.5cm,height=6cm]{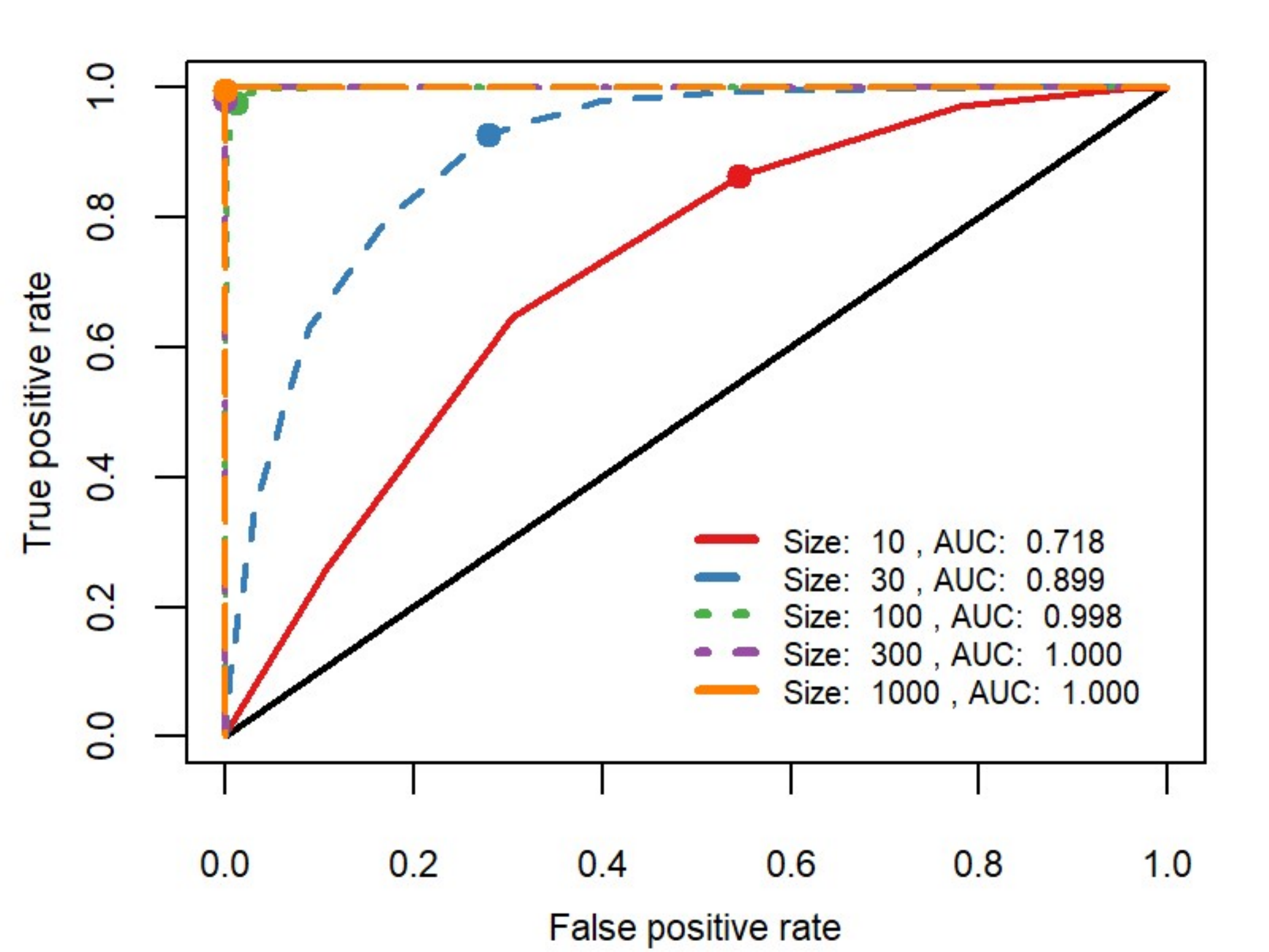} 
\caption{{\small ROC curve under complex graph}\label{fig:ROC_comp_str}} 
\end{subfigure}
\vskip\baselineskip
\begin{subfigure}[b]{0.45\textwidth}
\centering\includegraphics[width=7.5cm,height=6cm]{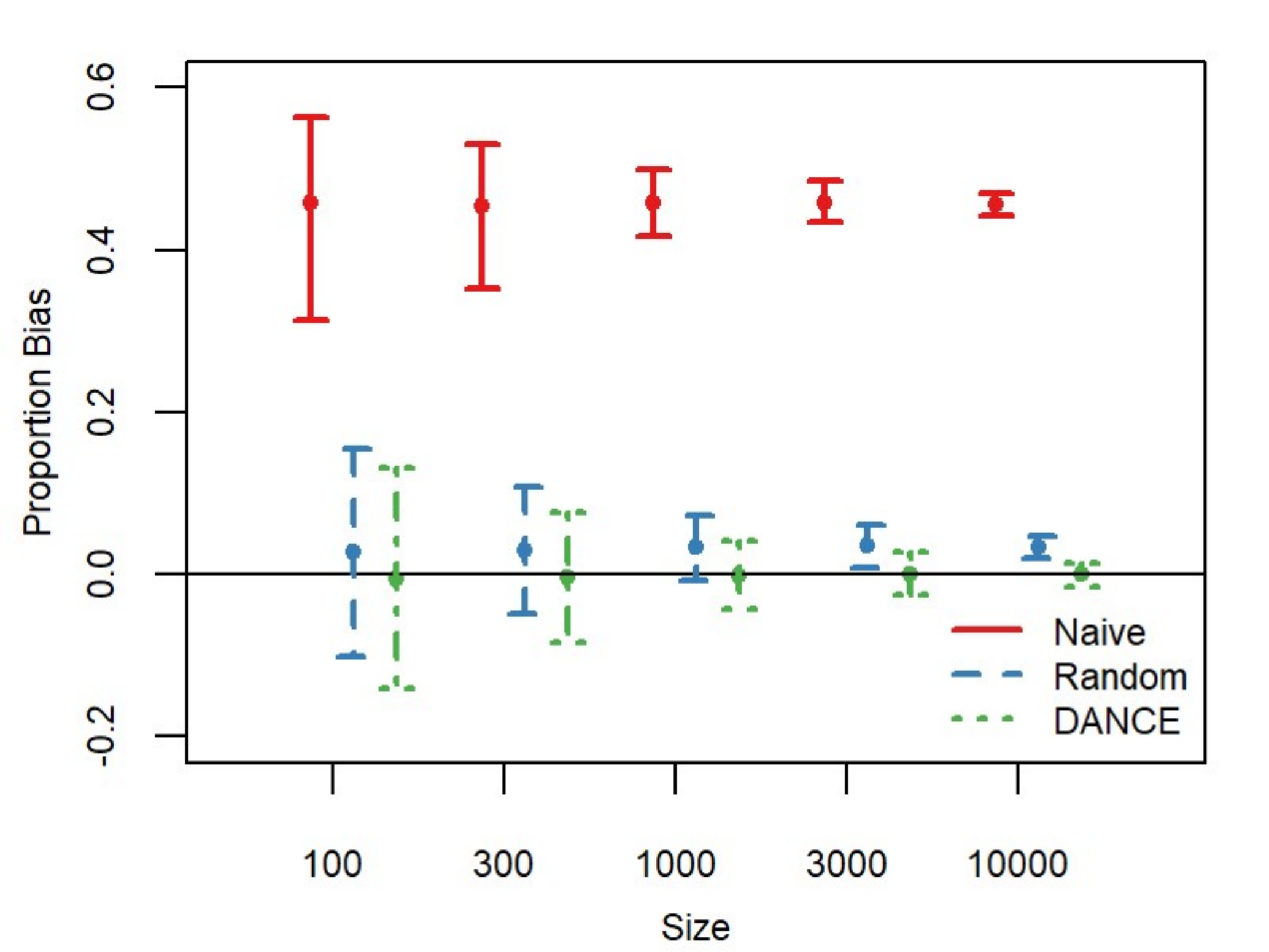}
\caption{{\small Proportion bias under simple graph}\label{fig:bias_simp_str}} 
\end{subfigure}
\hspace{0.2in}
\begin{subfigure}[b]{0.45\textwidth}
\centering\includegraphics[width=7.5cm,height=6cm]{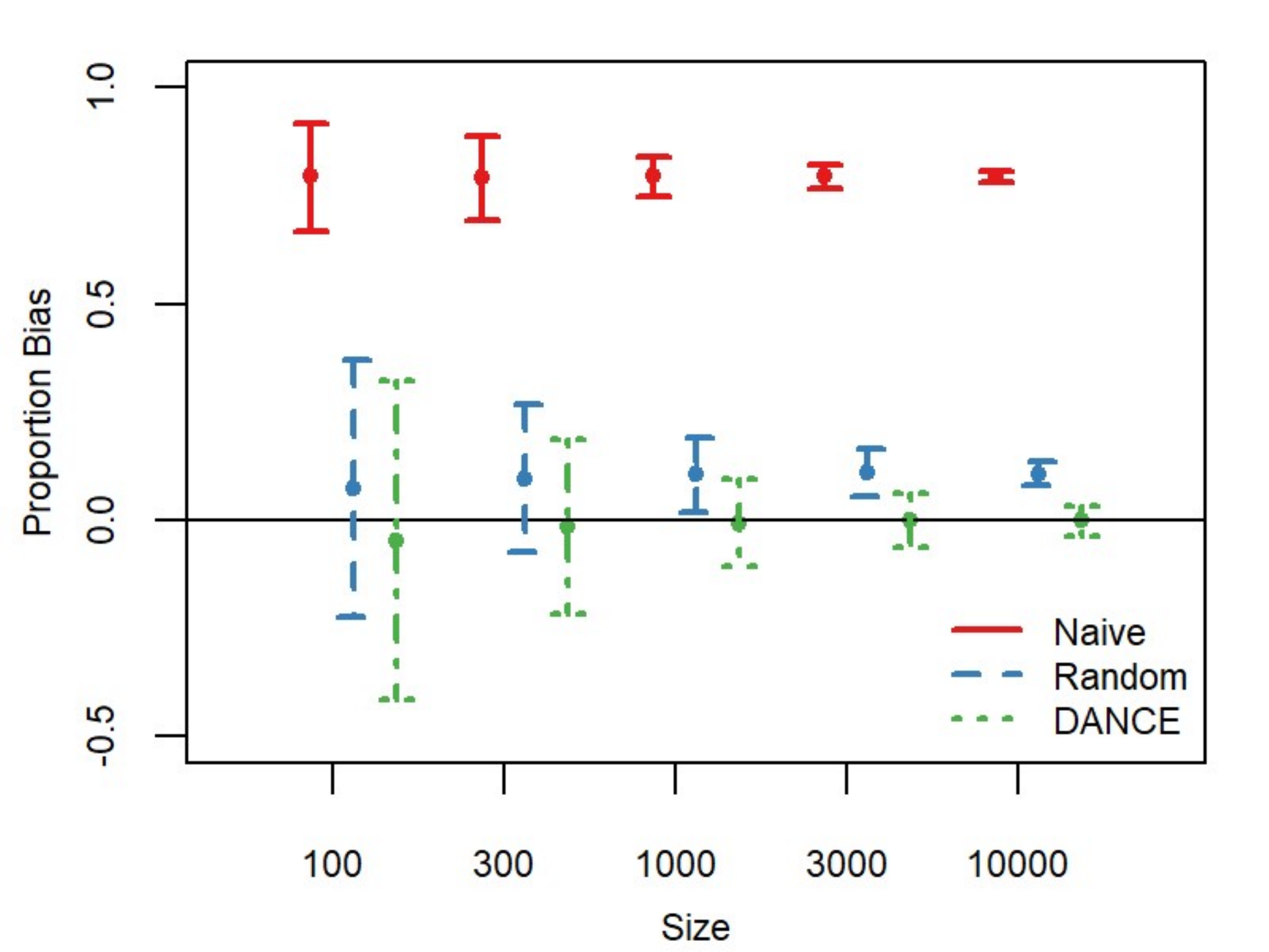}
\caption{{\small Proportion bias under complex graph}\label{fig:bias_comp_str}} 
\end{subfigure}
		\caption{Simulation results with data generated under simple and complex graphical structures with stronger edge strength.}
		\label{fig:SimresGraph2}
\end{figure}

\begin{table}[!htp]
	    \centering
	    \resizebox{\textwidth}{!}{\begin{tabular}{cccccccccccc}
\hline 
\multirow{2}{*}{Graph} & \multirow{2}{*}{Method} & \multicolumn{2}{c}{Bias} & \multicolumn{2}{c}{Proportion} & \multicolumn{2}{c}{Monte Carlo} & \multicolumn{2}{c}{Estimated} & \multicolumn{2}{c}{95\% CI}\tabularnewline
 &  & \multicolumn{2}{c}{($10^{-3}$)} & \multicolumn{2}{c}{Bias (\%)} & \multicolumn{2}{c}{SE ($10^{-3}$)} & \multicolumn{2}{c}{SE ($10^{-3}$)} & \multicolumn{2}{c}{Coverage}\tabularnewline
\hline 
\multicolumn{2}{c}{Sample size} & 1000 & 3000 & 1000 & 3000 & 1000 & 3000 & 1000 & 3000 & 1000 & 3000\tabularnewline
\hline 
\hline 
\multirow{2}{*}{Simple} & Random & 28.95 & 30.44 & 3.29 & 3.46 & 17.98 & 11.53 & 18.29 & 11.81 & 0.64 & 0.27\tabularnewline
 & DANCE & -1.08 & 0.49 & -0.12 & 0.06 & 18.99 & 11.98 & 18.75 & 12.28 & 0.94 & 0.96\tabularnewline
 &  &  &  &  &  &  &  &  &  &  & \tabularnewline
\multirow{2}{*}{Complex} & Random & 92.50 & 97.10 & 10.51 & 11.03 & 38.20 & 24.15 & 36.95 & 23.71 & 0.28 & 0.03\tabularnewline
 & DANCE & -4.98 & 1.12 & -0.57 & 0.13 & 45.72 & 28.10 & 42.90 & 28.05 & 0.95 & 0.96\tabularnewline
\hline 
\end{tabular}}
	    \caption{Operation Characteristics of the estimator based on randomly selected negative controls (Random) and the DANCE estimator for the simple and complex graphs, using stronger edge strengths.}
	    \label{tab:simu2}
	\end{table}

\begin{figure}[!t]
	\centering
\begin{subfigure}[b]{0.45\textwidth}
\centering\includegraphics[width=7.5cm,height=6cm]{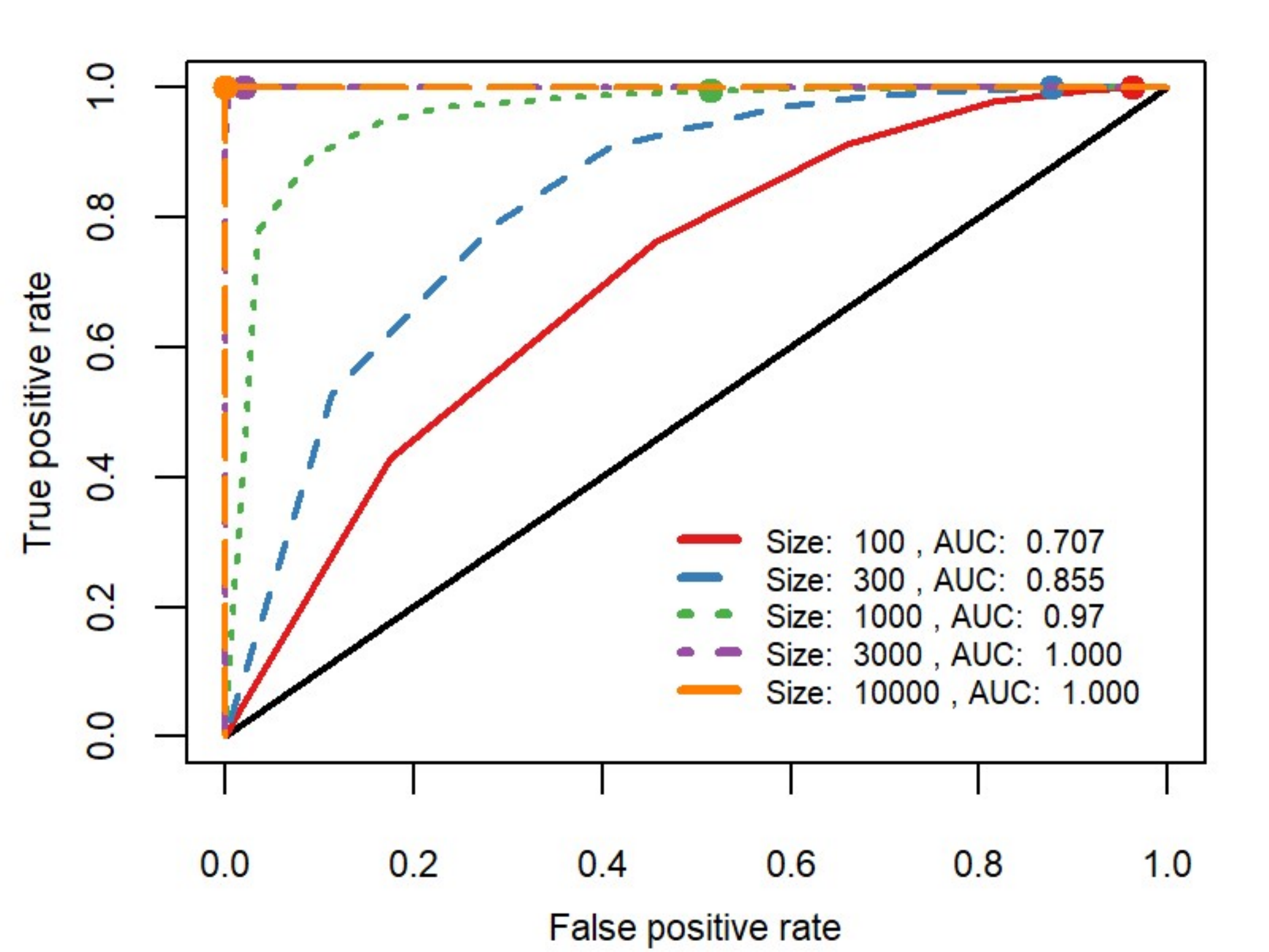}
\caption{{\small ROC curve under simple graph}\label{fig:ROC_simp_bin}} 
\end{subfigure}
\hspace{0.2in}
\begin{subfigure}[b]{0.45\textwidth}
\centering\includegraphics[width=7.5cm,height=6cm]{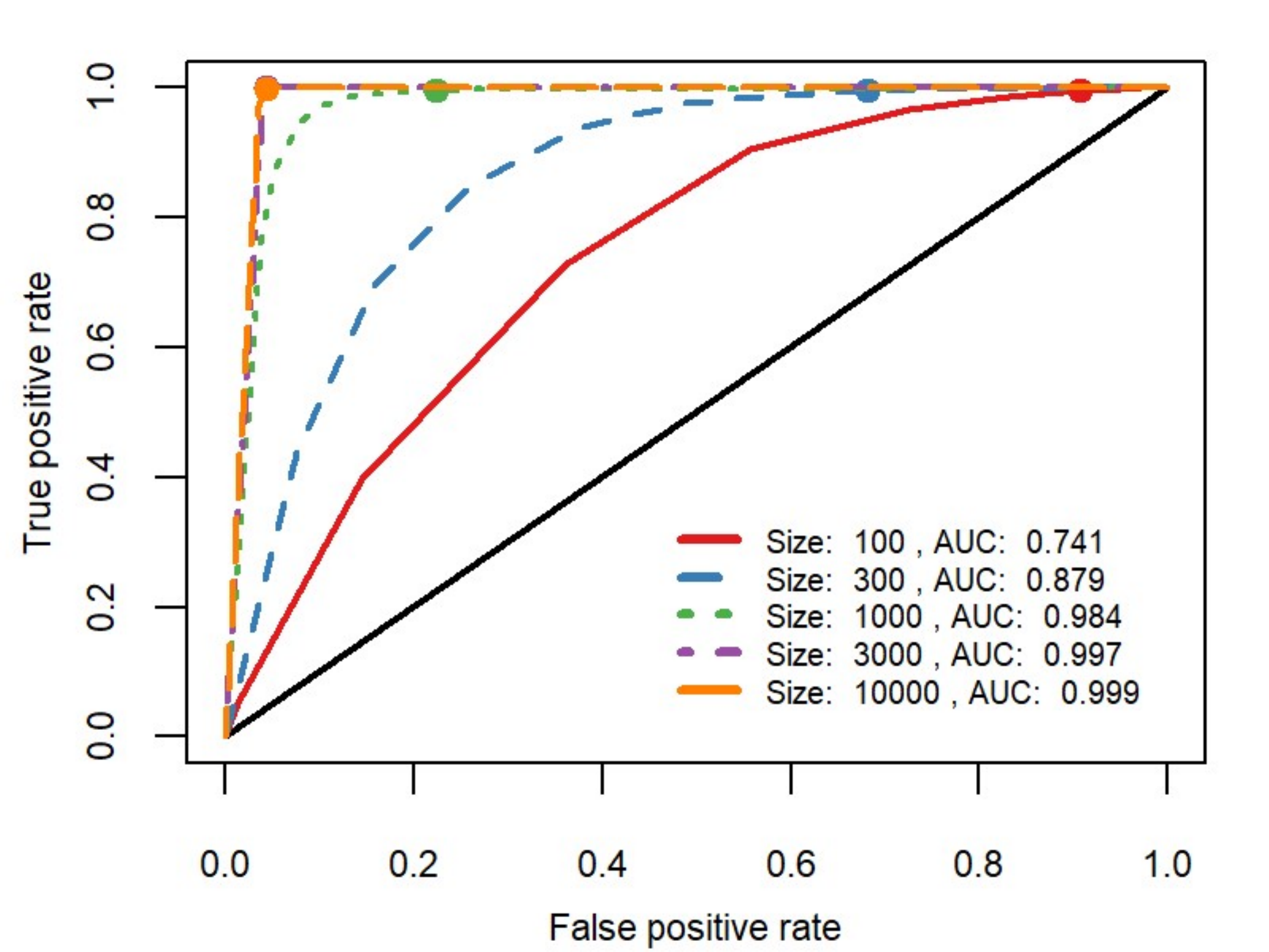} 
\caption{{\small ROC curve under complex graph}\label{fig:ROC_comp_bin}} 
\end{subfigure}
\vskip\baselineskip
\begin{subfigure}[b]{0.45\textwidth}
\centering\includegraphics[width=7.5cm,height=6cm]{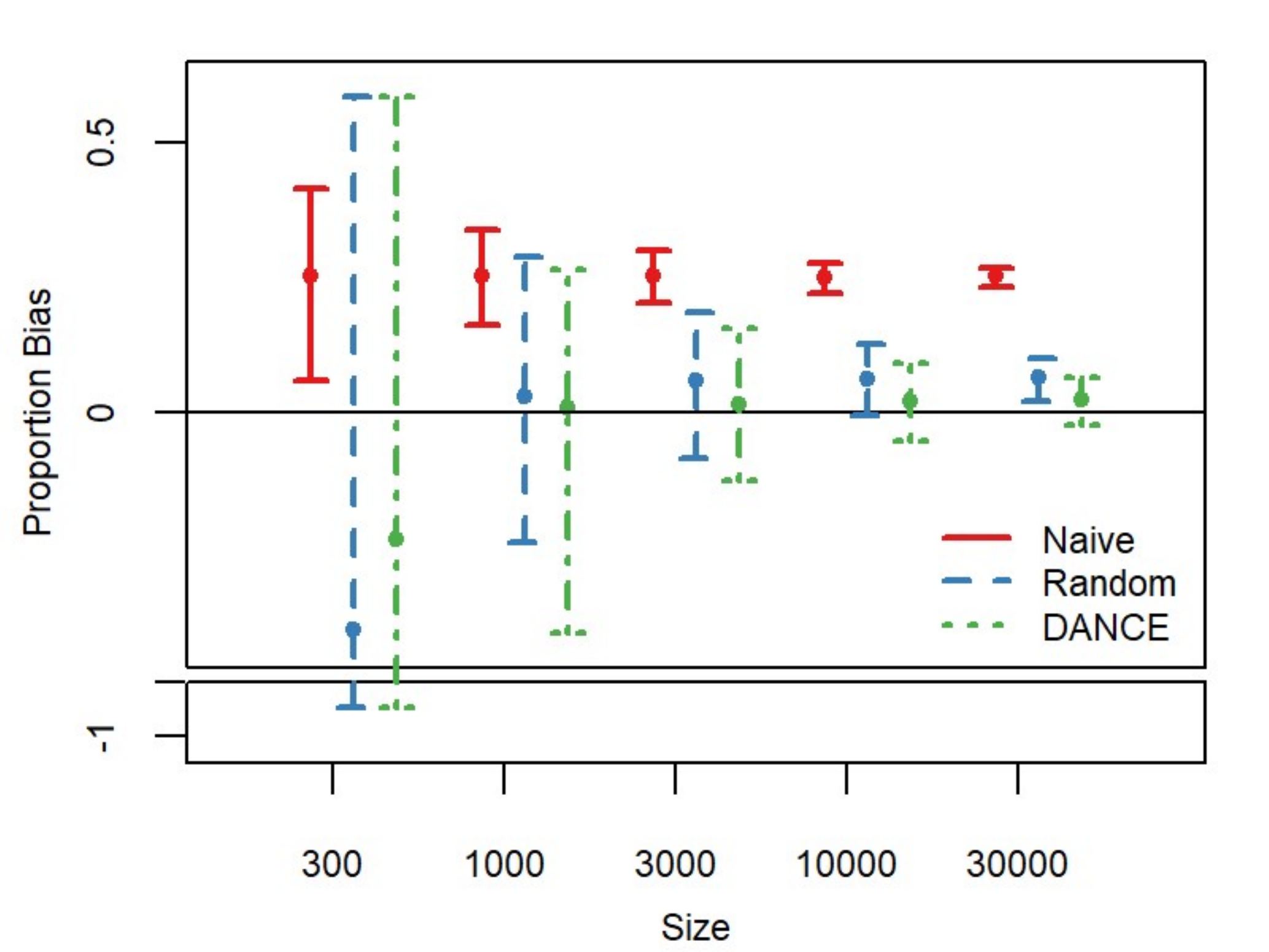}
\caption{{\small Proportion bias under simple graph}\label{fig:bias_simp_bin}} 
\end{subfigure}
\hspace{0.2in}
\begin{subfigure}[b]{0.45\textwidth}
\centering\includegraphics[width=7.5cm,height=6cm]{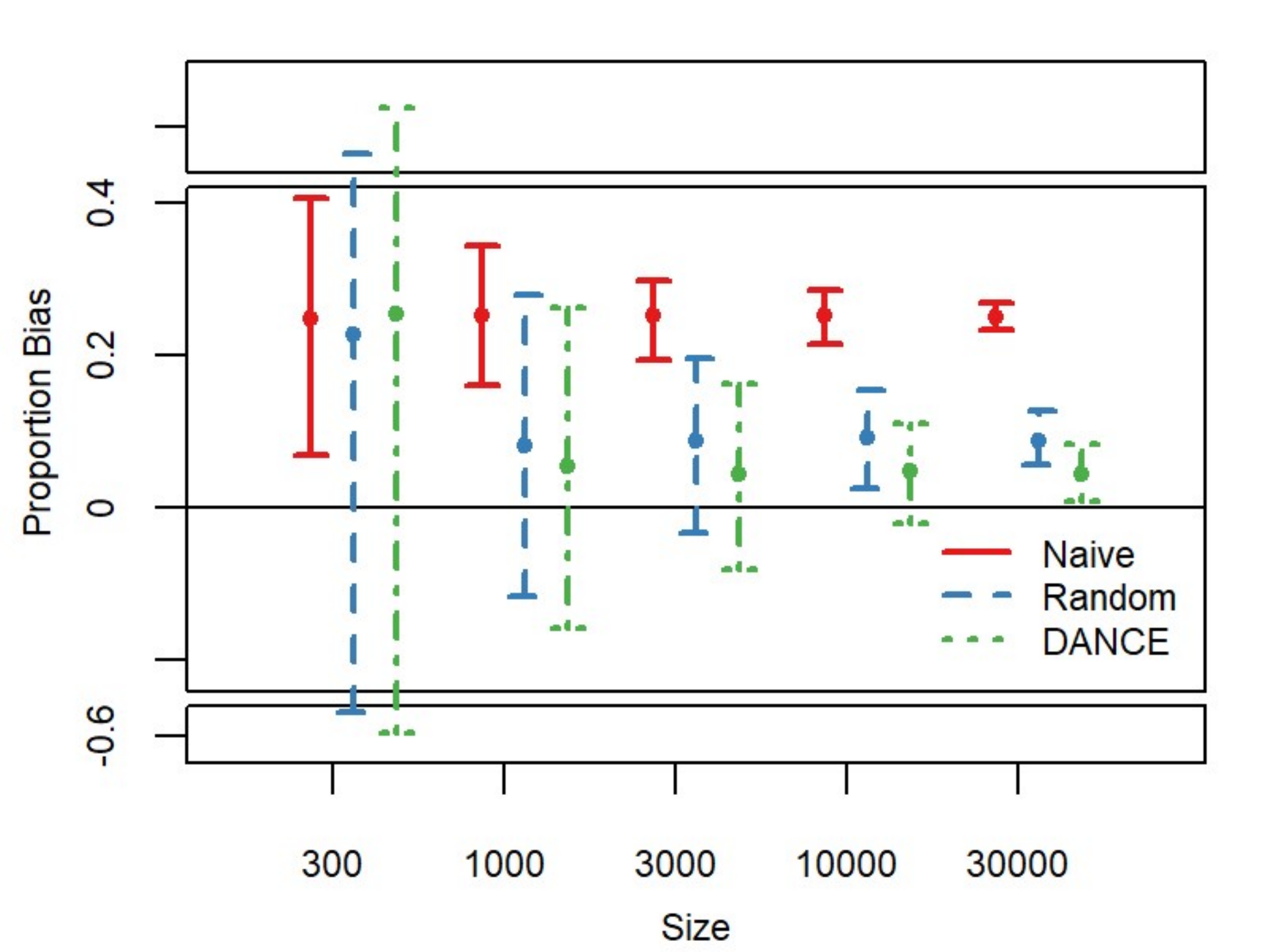}
\caption{{\small Proportion bias under complex graph}\label{fig:bias_comp_bin}} 
\end{subfigure}
		\caption{Simulation results with data generated under simple and complex graphical structure from binary random variables.}
		\label{fig:SimresGraph3}
\end{figure}

Figure~\ref{fig:SimresGraph3} shows the ROC curves for validation of negative controls and the proportion of bias for estimation of the ATE when all random variables generated are binary. As the ROC curve shows, the DANCE algorithm performs well under binary cases under both simple and complex graphical structures. DANCE's estimation of ATE is uniformly better than the Naive and Random methods.

	\clearpage
	\spacingset{1.45}
	\bibliographystyle{agsm}
	\bibliography{refs}

\end{document}